\documentclass[a4paper,11pt,dvipsnames]{article}
\usepackage[titles]{tocloft}

\usepackage{fullpage}
\usepackage{amsthm}
\usepackage{mathtools}
\usepackage{amssymb}
\usepackage{thmtools,thm-restate}
\usepackage{setspace}
\usepackage[hidelinks, pagebackref]{hyperref}
\usepackage{mathpazo}
\usepackage{enumerate}
\usepackage{bookmark}
\usepackage{appendix}
\usepackage{enumitem}
\usepackage{hyperref}
\usepackage{tikz}
\usepackage{algorithm, algorithmic}
\usepackage[misc]{ifsym}
   \date{}

\newcommand{\nocontentsline}[3]{}
\newcommand{\toclesslab}[3]{\bgroup\let\addcontentsline=\nocontentsline#1{#2\label{#3}}\egroup}

\hypersetup{
    colorlinks,
    linkcolor=[rgb]{0.3,0.3,0.6},
    citecolor=[rgb]{0.2, 0.6, 0.2},
    urlcolor=[rgb]{0.6, 0.2, 0.2}
}

\newcommand{\IP}{\mathsf{IP}}
\newcommand{\PSPACE}{\mathsf{PSPACE}}
\renewcommand{\P}{\mathsf{P}}
\newcommand{\NP}{\mathsf{NP}}
\newcommand{\QP}{\mathsf{QP}}
\newcommand{\poly}{\mathsf{poly}}

\newcommand{\DiDI}{\mathsf{DiDI}}

\newcommand{\F}{\mathbb{F}}
\newcommand{\N}{\mathbb{N}}

\newcommand{\Q}{\mathbb{Q}}
\newcommand{\Z}{\mathbb{Z}}
\newcommand{\C}{\mathbb{C}}

\newcommand{\ringR}{\mathsf{R}}

\newcommand{\x}{{\boldsymbol x}}
\renewcommand{\a}{\boldsymbol \alpha}
\newcommand{\T}{{\boldsymbol T}}
\newcommand{\y}{{\boldsymbol y}}
\newcommand{\z}{{\boldsymbol z}}

\newcommand{\e}{{\boldsymbol e}}

\newcommand{\calC}{\mathcal{C}}
\newcommand{\calD}{\mathcal{D}}
\newcommand{\calH}{\mathcal{H}}

\newcommand{\calT}{\mathcal{T}}

\newcommand{\calJ}{\mathcal{J}}

\renewcommand{\v}{\mathsf{val}}

\renewcommand{\sp}{\mathsf{sp}}

\newcommand{\rk}{\mathsf{rk}}
\newcommand{\cf}{\mathsf{coef}}
\renewcommand{\d}{\mathsf{deg}}
\newcommand{\dlg}{\mathsf{dlog}}
\newcommand{\size}{\mathsf{size}}
\newcommand{\tdeg}{\mathsf{trdeg}}
\newcommand{\PSP}[2]{\Pi^{#1}\Sigma\Pi^{#2}}
\newcommand{\SPSP}[3]{\Sigma^{#1}\Pi^{#2}\Sigma\Pi^{#3}}

\newcommand{\WSP}[2]{\wedge^{#1}\Sigma\Pi^{#2}}
\newcommand{\SWSP}[3]{\Sigma^{#1}\!\wedge^{#2}\!\Sigma\Pi^{#3}}
\newcommand{\PSU}[2]{\Pi^{#1} \Sigma\wedge^{#2}}
\newcommand{\SPSU}[3]{\Sigma^{#1} \Pi^{#2} \Sigma \wedge^{#3}}
\newcommand{\WSU}[1]{\wedge\Sigma \wedge^{#1}}
\newcommand{\SWSU}[2]{\Sigma^{#1}\!\wedge\!\Sigma\!\wedge^{#2}}
\newcommand{\Gen}[2]{\mathsf{Gen}(#1,#2)}
\newcommand{\SPS}[2]{\Sigma^{#1} \Pi^{#2} \Sigma}

\numberwithin{equation}{section}

\declaretheoremstyle[bodyfont=\normalfont \itshape \upshape,qed=\qedsymbol]{noproofstyle}

\declaretheoremstyle[bodyfont=\normalfont \itshape]{thmstyle}

\declaretheoremstyle[
    postheadhook = {\hspace*{\parindent}}]{thmliststyle}

\declaretheorem[numberlike=equation,style=thmstyle]{theorem}

\declaretheorem[name=Theorem,style=thmstyle,numbered=no]{theorem*}

\declaretheorem[numberlike=equation,style=thmstyle]{lemma}
\declaretheorem[name=Lemma,numbered=no,style=thmstyle]{lemma*}

\declaretheorem[name=Corollary,numbered=no,style=thmstyle]{corollary*}

\declaretheorem[name=Proposition,numbered=no,style=thmstyle]{proposition*}

\declaretheorem[numberlike=equation,style=thmstyle]{claim}
\declaretheorem[name=Claim,numbered=no,style=thmstyle]{claim*}

\declaretheorem[numberlike=equation,style=thmstyle]{fact}
\declaretheorem[name=Fact,numbered=no,style=thmstyle]{fact*}

\declaretheorem[numberlike=equation,style=thmstyle]{problem}
\declaretheorem[name=Problem,numbered=no,style=thmstyle]{problem*}

\declaretheorem[name=Conjecture,numbered=no,style=thmstyle]{conjecture*}

\declaretheorem[numberlike=equation,style=thmstyle]{definition}
\declaretheorem[unnumbered,name=Definition,style=thmstyle]{definition*}

\declaretheoremstyle[]{defstyle}

\declaretheorem[unnumbered,name=Remark,style=defstyle]{remark*}
\declaretheorem[unnumbered,style=thmliststyle, name=Remark]{remarklist*}

{\vspace{1ex}\noindent{\em Proof.}\hspace{0.5em}\def\myproof@name{#1}}%
{\hfill{\tiny \qed\ (\myproof@name)}\vspace{1ex}}
\newenvironment{proof-sketch}{\medskip\noindent{\em Proof Sketch.}}{\qed\bigskip}
\newenvironment{proof-attempt}{\medskip\noindent{\em Proof attempt.}}{\bigskip}

\usepackage{nth}
\usepackage{intcalc}
\usepackage{etoolbox}
\usepackage{xstring}

\usepackage{ifpdf}
\ifpdf
\else
\usepackage[quadpoints=false]{hypdvips}
\fi

\newcommand{\ECCCurl}[2]{http://eccc.hpi-web.de/report/\ifnumcomp{#1}{>}{93}{19}{20}#1/#2/}

\newcommand{\shortECCC}[2]{\texttt{\href{http://eccc.hpi-web.de/report/\ifnumcomp{#1}{>}{93}{19}{20}#1/#2/}{eccc:TR#1-#2}}}

\newcommand{\parseECCC}[1]{
\StrSubstitute{#1}{TR}{}[\tmpstring]%
\IfSubStr{\tmpstring}{/}{ 
\StrBefore{\tmpstring}{/}[\ecccyear]%
\StrBehind{\tmpstring}{/}[\ecccreport]%
}{
\StrBefore{\tmpstring}{-}[\ecccyear]%
\StrBehind{\tmpstring}{-}[\ecccreport]%
}%
\shortECCC{\ecccyear}{\ecccreport}}

\allowdisplaybreaks
\begin{document}

\title{Deterministic identity testing paradigms for bounded top-fanin depth-4 circuits\thanks{A preliminary version appeared in $36^{\text{th}}$ Computational Complexity Conference (CCC), 2021. \cite{DDS21-Didi}}}

\author{Pranjal Dutta\thanks{School of Computing, NUS. Email: \texttt{duttpranjal@gmail.com}} \and Prateek Dwivedi\thanks{Dept.~of Computer Science \& Engineering, IIT Kanpur. Email: \texttt{\{pdwivedi,nitin\}@cse.iitk.ac.in}} \and Nitin Saxena\footnotemark[2]}

\maketitle
\begin{abstract}
    Polynomial Identity Testing (PIT) is a fundamental computational problem. The famous depth-$4$ reduction result by Agrawal and Vinay (FOCS 2008) has made PIT for depth-$4$ circuits an enticing pursuit. A restricted depth-4 circuit computing a $n$-variate degree-$d$ polynomial of the form $\sum_{i = 1}^{k} \prod_{j} g_{ij}$, where $\deg g_{ij} \leq \delta$ is called $\SPSP{[k]}{}{[\delta]}$ circuit. On further restricting $g_{ij}$ to be sum of univariates we obtain $\SPSU{[k]}{}{}$ circuits. The largely open, special-cases of $\SPSP{[k]}{}{[\delta]}$ for constant $k$ and $\delta$, and $\SPSU{[k]}{}{}$ have been a source of many great ideas in the last two decades. For eg. depth-$3$ ideas of Dvir and Shpilka (STOC 2005), Kayal and Saxena (CCC 2006), and Saxena and Seshadhri (FOCS 2010 and STOC 2011). Further, depth-$4$ ideas of Beecken, Mittmann and Saxena (ICALP 2011), Saha, Saxena and Saptharishi (Comput.Compl. 2013), Forbes (FOCS 2015), and Kumar and Saraf (CCC 2016). Additionally, geometric Sylvester-Gallai ideas of Kayal and Saraf (FOCS 2009), Shpilka (STOC 2019), and Peleg and Shpilka (CCC 2020, STOC 2021). Very recently, a subexponential-time {\em blackbox} PIT algorithm for constant-depth circuits was obtained via lower bound breakthrough of Limaye, Srinivasan, Tavenas (FOCS 2021). We solve two of the basic underlying open problems in this work.

    We give the \emph{first} polynomial-time PIT for $\SPSU{[k]}{}{}$. We also give the \emph{first} quasipolynomial time {\em blackbox} PIT for both $\SPSU{[k]}{}{}$ and $\SPSP{[k]}{}{[\delta]}$. A key technical ingredient in all the three algorithms is how the {\em logarithmic derivative}, and its power-series, modify the top $\Pi$-gate to $\wedge$.
\end{abstract}

\vspace*{1em}

\noindent \textbf{2012 ACM Subject Classification } Theory of computation \(\to\) Algebraic complexity theory.
\vspace*{1em}

\noindent \textbf{Acknowledgement} Pranjal is supported by the project “Foundation of Lattice-based Cryptography", funded by NUS-NCS Joint Laboratory for Cyber Security. Most of the work was carried out when Pranjal was visiting CSE, IIT Kanpur and supported by Google PhD  Fellowship. Nitin thanks the funding support from DST (SJF/MSA-01/2013-14), SERB (CRG/2020/000045) and N.Rama Rao Chair.


\tableofcontents

\section{Introduction: PIT \& beyond}

Algebraic circuits are natural algebraic analog of boolean circuits, with the logical operations being replaced by $+$ and $\times$ operations over the underlying field. The study of algebraic circuits comprise the large study of algebraic complexity, mainly pioneered (and formalized) by Valiant \cite{valiant1979completeness}. A central problem in algebraic complexity is an algorithmic design problem, known as Polynomial Identity Testing (PIT): given an algebraic circuit $\calC$ over a field $\F$ and input variables $x_1,\hdots,x_n$, determine whether $\calC$ computes the identically zero polynomial. PIT has found numerous applications and connections to other algorithmic problems. Among the examples are algorithms for finding perfect matchings in graphs \cite{lovasz1979determinants, DBLP:journals/combinatorica/MulmuleyVV87,fenner2019bipartite}, primality testing \cite{agrawal2004primes}, polynomial factoring \cite{kopparty2014equivalence,dutta2018discovering}, polynomial equivalence \cite{dvir2014testing}, reconstruction algorithms \cite{klivans2006learning,shpilka2009interpolation,karnin2009reconstruction} and the existence of algebraic natural proofs \cite{chatterjee2020existence,kumar2020if}. Moreover, efficient design of PIT algorithms is intrinsically connected to proving strong lower bounds \cite{heintz1980testing,agrawal2005proving,kabanets2004derandomizing,dvir2010hardness,forbes2018succinct,chou2018hardness,dutta_et_al:LIPIcs.ITCS.2021.23}. Interestingly, PIT also emerges in many fundamental results in complexity theory such as $\IP=\PSPACE$ \cite{shamir1992ip,lund1992algebraic}, the PCP theorem \cite{arora1998proof,arora1998probabilistic}, and the overarching Geometric Complexity Theory (GCT) program towards $\P\ne\NP$ \cite{mulmuley2012gct,mulmuley2012geometric, grochow2015unifying,joshua16boundaries}. 

There are broadly two settings in which the PIT question can be framed. In the \emph{whitebox} setup, we are allowed to look inside the wirings of the circuit, while in the \emph{blackbox} setting we can only evaluate the circuit at some points from the given domain. There is a very simple randomized algorithm for this problem - evaluate the polynomial at a random point from a large enough domain. With very high probability, a nonzero polynomial will have a nonzero evaluation; this is famously known as the Polynomial Identity Lemma \cite{ore1922hohere,DL78,Zip79,schwartz1980fast}. It has been a
long standing open question to derandomize this algorithm.

For many years, blackbox identity tests were only known for depth-2 circuits which compute sparse polynomials \cite{ben1988deterministic,klivans2001randomness}. In a surprising result, Agrawal and Vinay \cite{agrawal2008arithmetic} showed that a complete derandomization of blackbox identity testing for just depth-4 algebraic circuits ($\SPSP{}{}{}$) already implies a near complete derandomization for the general PIT problem. More recent depth reduction results \cite{koiran2012arithmetic,gupta2013arithmetic}, and the bootstrapping phenomenon \cite{Agrawal8107,kumar2019near,DBLP:conf/focs/Guo0SS19,DBLP:conf/coco/000320} show that even PIT
for very restricted classes of depth-$4$  circuits (\emph{even} depth-$3$) would have very
interesting consequences for PIT of general circuits. These results make the identity testing regime for depth-$4$ circuits, a very meaningful pursuit.

\emph{Three PITs in one-shot.} Following the same spirit, here we solve three important (and open) PIT questions. We give the first deterministic polynomial-time whitebox PIT algorithm for the bounded sum of product of sum of univariates circuits \cite[Open Prob.~2]{saha2013case}. Further, we give a quasipolynomial-time blackbox algorithm for the same class of circuits. These circuits are denoted by $\SPSU{[k]}{}{}$ and compute polynomials of the form $\Sigma_{i \in [k]} \Pi_{j} \left(g_{ij1}(x_1)+ \cdots +g_{ijn}(x_n)\right)$.
\begin{quote}
    \itshape 
    Whitebox and Blackbox PIT for the $\SPSU{[k]}{}{}$ circuits is in polynomial and quasi-polynomial time respectively.
\end{quote}
A similar technique also gives a quasi-polynomial time blackbox PIT algorithm for the bounded sum of product of bounded degree sparse polynomials circuits. They are denoted by $\SPSP{[k]}{}{[\delta]}$ (where $k$ and $\delta$ can be up to \(\poly(\log(s))\), where \(s\) is the circuit size). 
\begin{quote}
    \itshape
    Blackbox PIT for the $\SPSP{[k]}{}{[\delta]}$ circuits is in quasi-polynomial time.
\end{quote}
$\SPSP{[k]}{}{[\delta]}$ circuits compute polynomials which are of the form $\Sigma_{i \in [k]} \Pi_{j} g_{ij}(\x)$, where $\d(g_{ij}) \le \delta$. Even $\delta=2$ was a challenging open problem \cite[Open Problem~2]{DBLP:conf/coco/0001S16a}. The model has gained a lot of interest in the past few years and has generated many important results \cite{peleg2020generalized,peleg2020polynomial,GOS22,OS22}.

\vspace*{1em}
\subsection{Main results: An analytic detour to three PITs}
Though some attempts have been made to solve PIT for $\SPSU{[k]}{}{}$, an efficient PIT for $k \ge 3$  \emph{even} in the whitebox settings remains open, see \cite[Open Prob.~2]{saha2013case}. Our first result addresses this problem and designs a polynomial time algorithm (Ref. Algorithm \ref{algo:whitebox-algo}). In our pursuit we discover an analytic and non-ideal based new technique which we refer as $\mathsf{DiDI}$. Throughout the paper, we will work with $\F=\Q$, though all the results hold for field of large characteristic. 

\begin{theorem}[Whitebox $\SPSU{[k]}{}{}$ PIT] \label{thm:thm1}
There is a deterministic, whitebox $s^{O(k\,7^k)}$-time PIT algorithm for $\SPSU{[k]}{}{}$ circuits of size $s$, over $\F[\x]$.
\end{theorem}

\begin{remark*} ~ 
    \begin{enumerate}
        \item Case $k \le 2$ can be solved by invoking \cite[Theorem 5.2]{saha2013case}; but $k \ge3$ was open.
        \item Our technique \emph{necessarily} blows up the exponent exponentially in $k$. In particular, it would be interesting to design an efficient time algorithm when $k=\Theta(\log s)$.
        \item It is not clear if the current technique gives PIT for $\Sigma^{[k]}\Pi\Sigma\mathsf{M}_2$ circuits, where \(\Sigma\mathsf{M}_2\) denotes sum of {\em bi}variate monomials computed and fed into the top product gate.
    \end{enumerate}
\end{remark*}

Next, we go to the blackbox setting and address two models of interest, namely--- $\SPSU{[k]}{}{}$ and $\SPSP{[k]}{}{[\delta]}$, where $k, \delta$ are constants. Our work builds on previous ideas for unbounded top fanin (1) Jacobian \cite{agrawal2016jacobian}, (2) the known blackbox PIT for $\SWSU{}{}{}$ and $\SWSP{}{}{[\delta]}$ \cite{gurjarARO,forbes2015deterministic} while maneuvering with an analytic approach \emph{via} power-series, which unexpectedly \emph{reduces} the top $\Pi$-gate to a $\wedge$-gate.

\begin{theorem}[Blackbox depth-$4$ PIT]\label{thm:thm2}
    \textcolor{white}{.}
    \begin{enumerate}[label={(\alph*)}]
        \item There is a $s^{O(k \log \log s)}$ time blackbox PIT algorithm for $\SPSU{[k]}{}{}$ circuits of size $s$, over $\F[\x]$.
        \item There is a $s^{O(\delta^2\,k\,\log s)}$ time blackbox PIT algorithm for $\SPSP{[k]}{}{[\delta]}$ circuits of size $s$, over $\F[\x]$.
    \end{enumerate}
\end{theorem}

\begin{remark*} ~ 
    \begin{enumerate}
    \setlength\itemsep{.1em}
        \item \autoref{thm:thm2} (b) has a \emph{better} dependence on $k$, but \emph{worse} on $s$, than \autoref{thm:thm1}. Our results are quasipoly-time even up to $k,\delta = \poly(\log s)$.
        \item \autoref{thm:thm2} (a) is better than \autoref{thm:thm2} (b), because $\SWSU{}{}{}$ has a faster algorithm than $\SWSP{}{}{[\delta]}$. 
        \item Even for $\SPSU{[3]}{}{}$ and $\SPSP{[3]}{}{[3]}$ models, we leave the {\em poly}-time blackbox question open. 
    \end{enumerate}
\end{remark*}

\subsection{Prior works on related models}  

In the last two decades, there has been a surge of results on identity testing for restricted
classes of bounded depth algebraic circuits (e.g.~`locally' bounded independence, bounded read/occur, bounded variables). There have been numerous results on PIT for depth-3 circuits with bounded
top fanin (known as $\Sigma^{[k]}\Pi\Sigma$-circuits). Dvir and Shpilka \cite{dvir2007locally} gave the first quasipolynomial-time deterministic whitebox algorithm for $k=O(1)$, using rank based methods, which finally lead Karnin and Shpilka \cite{karnin2011black} to design algorithm of same complexity in the blackbox setting. Kayal and Saxena \cite{kayal2007polynomial} gave the first polynomial-time algorithm of the same. Later, a series of works 
in \cite{saxena2011almost,saxena2012blackbox,saxena2013sylvester, agrawal2016jacobian} generalized the model and gave $n^{O(k)}$-time algorithm when the algebraic rank of the product polynomials are bounded. Note that in the white-box setting, our algorithm gives a poly(s) time PIT algorithm for bounded top-fanin depth-3 circuit. Moreover, the dependence on the top-fan is exponential. In the blackbox setting, our algorithm solves PIT for bounded top-fanin depth-3 circuit in quasi-poly(s) time, hence it does not offer any speedup compared to known polynomial time algorithms. However, our algorithm does give a PIT idea that is different from the known ones.

There has also been some progress on PIT for restricted classes of depth-4 circuits. A quasipolynomial-time blackbox PIT algorithm for \emph{multilinear} $\SPSP{[k]}{}{}$-circuits was designed in \cite{karnin2013deterministic}, which was further improved to a $n^{O(k^2)}$-time deterministic algorithm in \cite{saraf2018black}. A quasipolynomial blackbox PIT was given in \cite{beecken2013algebraic,DBLP:conf/coco/0001S16a} when algebraic rank of the irreducible factors in each multiplication gate as well as the bottom fanin are bounded. Further interesting restrictions like sum of product of fewer variables, and more structural restrictions have been exploited, see \cite{forbes2013quasipolynomial,agrawal2013quasi, forbes2015deterministic,mukhopadhyay2016depth,DBLP:journals/toc/0001S17}. Some progress has also been made for bounded top-fanin and bottom-fanin depth-$4$ circuits via incidence geometry \cite{gupta2014algebraic,shpilka2019sylvester,peleg2020generalized}. In fact, very recently, \cite{peleg2020polynomial} gave a polynomial-time blackbox PIT for $\SPSP{[3]}{}{[2]}$-circuits.

\begin{table}[ht]
    \centering
        \renewcommand{\arraystretch}{1.5}
        \begin{tikzpicture}[line width=.1mm,rounded corners=.5em]
        \node (table) [clip,inner sep=.5\pgflinewidth] {
        \begin{tabular}{llll}
            Model & Time & Ref. \\ \hline
            $\SPS{[k]}{[d]}$ & $\poly(n,d^k)$ & \cite{saxena2012blackbox} \\
            Multilinear $\SPSP{[k]}{}{}$ & $\poly(n^{O(k^2)})$ & \cite{saraf2018black,agrawal2016jacobian}\\
            $\SPSP{}{}{}$ of bounded $\mathsf{trdeg}$ & $\poly(s^{\mathsf{trdeg}})$ & \cite{beecken2013algebraic}\\ 
            $\SPSP{(k)}{}{[d]}$ of bounded \emph{local} $\mathsf{trdeg}$ & $\mathsf{QP}(n)$ & \cite{DBLP:journals/toc/0001S17} \\ 
            $\SPSP{[3]}{}{[2]}$ & $\poly(n,d)$ & \cite{peleg2020polynomial}\\
            $\overline{\SPSU{[k]}{}{}}$ & $s^{O(k \cdot 7^k \cdot \log log s)}$ & \cite{DDS21} \\
            $\overline{\SPSP{[k]}{}{[\delta]}}$ & $s^{O(\delta^2 \cdot k \cdot 7^k \cdot \log s)}$ & \cite{DDS21}\\
            $\SPSP{}{}{}$ & SUBEXP(n) & \cite{LST21}\\
            Whitebox $\SPSU{[k]}{}{}$ & $s^{O(k\,7^k)}$ & This work.\\
            $\SPSU{[k]}{}{}$ & $s^{O(k \log \log s)}$ & This work.\\
            $\SPSP{[k]}{}{[\delta]}$ & $s^{O(\delta^2\,k\,\log s)}$ & This work.\\
        \end{tabular}
        };
        \draw ([xshift=.5*\pgflinewidth,yshift=-.5*\pgflinewidth]table.north west) 
        rectangle ([xshift=-.5*\pgflinewidth,yshift=.5*\pgflinewidth]table.south east);
        \end{tikzpicture}
        \caption{Time complexity comparision of PIT algorithms related to $\SPSP{}{}{}$ circuits}
\end{table}

The authors recently generalised their novel $\DiDI$-technique to solve 'border PIT' of depth-4 circuits \cite{DDS21}. Specifically, they give a $s^{O(k \cdot 7^k \cdot \log log s)}$ time and $s^{O(\delta^2 \cdot k \cdot 7^k \cdot \log s)}$ time blackbox PIT algorithm for $\overline{\SPSU{[k]}{}{}}$ and $\overline{\SPSP{[k]}{}{[\delta]}}$ respectively. By definition, border classes capture exact complexity classes, hence border PIT results seemingly subsumes the results we present in this paper. However, the whitebox PIT algorithm here is much more efficient than their quasi-poly time blackbox algorithm. Further, the time complexity of blackbox PIT algorithms has a better dependence on $k$ and $\delta$ compared to their exponential dependence. Lastly, the proofs in this paper are simpler as we don't have to deal with an infinitesimally close approximation of polynomials in  border complexity classes. Very recently, Dutta and Saxena~\cite{duttaseparated21} showed an exponential-gap fanin-hierarchy theorem for bounded depth-3 circuits which is also based on a {\em finer} generalization of the $\DiDI$-technique.

In a breakthrought result by Limaye, Srinivasan and Tavenas \cite{LST21} the {\em first} superpolynomial lower bound for constant depth circuits was obtained. Their lower bound result, together with the `hardness vs randomness' tradeoff result of \cite{chou2018hardness} gives the {\em first} deterministic blackbox PIT algorithm for general depth-4 circuits which runs in $s^{O(n^\epsilon)}$ for all real $\epsilon > 0$. Their result is the first {\em sub}exponential time PIT algorithm for depth-4 circuits. Moreover, compared to their algorithm, our quasipoly time blackbox and polynomial time whitebox algorithms are significantly faster.

\medskip
{\bf Limitations of known techniques.}~~People have studied depth-$4$ PIT only with extra restrictions, mostly due to the limited applicability of the existing techniques as they were tailor-made for the specific models and do not generalize. E.g.~the previous methods handle  $\delta=1$ (i.e. linear polynomials at the bottom) or $k=2$ (via \emph{factoring}, \cite{saha2013case}). While $k=2$ to $3$, or $\delta=1$ to $2$ (i.e. `linear' to `quadratic') already demands a qualitatively different approach.  

Whitebox $\SPSU{[k]}{}{}$ model generalizes the famous bounded top fanin depth-$3$ circuits $\Sigma^{[k]}\Pi\Sigma$ of  \cite{kayal2007polynomial}; but their Chinese Remaindering (CR) method, loses applicability and thus fails to solve even a slightly more general model. The blackbox setting involved similar `certifying path' ideas in \cite{saxena2012blackbox} which could be thought of as general CR. It comes up with an ideal $I$ such that $f \ne 0 \bmod I$ and finally preserves it under a constant-variate linear map. The preservation gets harder (for both $\SPSU{[k]}{}{}$ and $\SPSP{[k]}{}{[\delta]}$) due to the increased non-linearity of the ideal $I$ generators. Intuitively, larger $\delta$ via ideal-based routes, brings us to the Gr\"obner basis method (which is doubly-exponential-time in $n$) \cite{vasconcelos2004computational}. We know that ideals even with $3$-generators (analogously $k=4$) already capture the whole ideal-membership problem \cite{sap19}.

The algebraic-geometric approach to tackle $\SPSP{[k]}{}{[\delta]}$ has been explored in \cite{beecken2013algebraic,gupta2014algebraic,mukhopadhyay2016depth,zeyu2021variety}. The families which satisfy a certain Sylvester–Gallai configuration (called SG-circuits) is the harder case which is conjectured to have constant transcendence degree \cite[Conj.~1]{gupta2014algebraic}. Non-SG circuits is the case where the nonzeroness-certifying-path question reduces to radical-ideal  non-membership questions \cite{garg2020special}. This is really a variety question where one could use algebraic-geometry tools to design a poly-time blackbox PIT. In fact, very recently, Guo \cite{zeyu2021variety} gave a $s^{\delta^k}$-time PIT by constructing explicit variety evasive subspace families. Unfortunately, this is not the case in the ideal non-membership; this scenario makes it much harder to solve $\SPSP{[k]}{}{[\delta]}$. From this viewpoint, radical-ideal-membership explains well why the intuitive $\Sigma^{[k]} \Pi \Sigma$ methods do not extend to  $\SPSP{[k]}{}{[\delta]}$.

Interestingly, Forbes \cite{forbes2015deterministic} found a quasipolynomial-time PIT for $\SWSP{}{}{[\delta]}$ using shifted-partial derivative techniques; but it naively fails when one replaces the $\wedge$-gate by $\Pi$ (because the `measure' becomes too large). The duality trick of \cite{saxena2008diagonal} completely solves whitebox PIT for $\SWSU{}{}$, by transforming it to a read-once oblivious ABP (ROABP); but it is inapplicable to our models with the top $\Pi$-gate (due to large waring rank and ROABP-width). A priori, our models are incomparable to ROABP, and thus the famous PIT algorithms for ROABP \cite{forbes2013quasipolynomial,forbes2014hitting,gurjarARO} are not expected to help either.

Similarly, a naive application of the \emph{Jacobian} and \emph{certifying path} technique from \cite{agrawal2016jacobian} fails for our models because it is difficult to come up with a \emph{faithful} map for constant-variate reduction. Kumar and Saraf \cite{DBLP:conf/coco/0001S16a} crucially used that the computed polynomial has low individual degree (such that~\cite{dvir2010hardness} can be invoked), while in \cite{DBLP:journals/toc/0001S17} they exploits the low algebraic rank of the polynomials computed below the top $\Pi$-gate. Neither of them hold in general for our models. Very recently, Peleg and Shpilka \cite{peleg2020polynomial} gave a poly-time blackbox PIT for $\SPSP{[3]}{}{[2]}$, via incidence geometry (e.g.~Edelstein-Kelly theorem involving `quadratic' polynomials), by solving \cite[Conj.~1]{gupta2014algebraic} for $k=3,\delta=2$. The method seems very strenuous to generalize even to `cubic' polynomials ($\delta=3=k$). 

\textbf{PIT for other models. } Blackbox PIT algorithms for many restricted models are known. Egs.~ROABP related models \cite{raz2005deterministic,jansen2010deterministic, agrawal2015hitting,gurjarARO,gurjar2017deterministic,forbes2014hitting,anderson2018identity}, $\log$-variate circuits \cite{forbes2018towards,bisht2020poly}, and non-commutative models \cite{garg2016deterministic,DBLP:journals/cjtcs/LagardeMP19}.

\subsection{Techniques and motivation}\label{sec:pf-ideas}

Both the proofs are analytic as they use \emph{logarithmic derivative}, and its power-series expansion which greatly transform the respective models. Where the nature of the first proof is inductive, the second is a more direct \emph{one-shot} proof. In both the cases, we essentially reduce to the well-understood \emph{wedge} models, that have unbounded top fanin, yet for which PITs are known. This reduction is unforeseeable and quite `power'ful.  

The analytic tool that we use, appears in algebra and complexity theory through the \emph{formal power series} ring $\ringR[[x_1,\hdots,x_n]]$ (in short $\ringR[[\x]]$), see \cite{10.2307/2317940,sinhababu2019power,dutta2018discovering}. The advantages of the ring $\ringR[[\x]]$ are many and they usually emerge because of the inverse identity: $(1 - x_1)^{-1} = \sum_{i \ge 0}\, x_1^i$, which does not make sense in $\ringR[x]$, but is valid in $\ringR[[\x]]$. Other analytic tools used are inspired from Wronskian (linear dependence) \cite[Theorem~7]{koiran2015wronskian} \cite{kayal2015lower}, Jacobian (algebraic dependence) \cite{beecken2013algebraic, agrawal2016jacobian,pandey2018algebraic}, and logarithmic derivative operator $\dlg_{\,z_1}(f) = (\partial_{z_1}\,f)/f$.

We will be work with the division operator (e.g.~ $\ringR(z_1)$, over a certain ring $\ringR$). However, the divisions do not come for free as they require invertibility with respect to $z_1$ throughout (again landing us in $\ringR[[z_1]]$. For circuit classes $C, D$ we define class \[\calC/\calD:=\{f/g \mid f\in \calC, \calD \ni g \neq 0\}.\]Similarly $\calC \cdot \calD$ to denotes the class taking respective products.

\subsubsection{The \texorpdfstring{$\DiDI$}{DiDI}-technique} In \autoref{thm:thm1} we introduce a novel technique for designing PIT algorithms which comprises of inductively applying two fundamental operations on the input circuits to reduce it to a more tractable model. Suppose we want to test $\sum_{i \in [k]} T_i \stackrel{?}{=} 0$ where each $T_i$ is computable by $\PSU{}{}$. The idea is to \emph{DI}vide it by $T_k$ to obtain $1 + \sum_{i \in [k-1]} T_i/T_k$ and then {\em D}erivative to reduce the fanin to $k-1$ and obtain $\sum_{i \in [k-1]} \calT_i$. Naturally, these operations pushes us to work with the fractional ring (e.g.~ $\ringR(z_1)$, over a certain ring $\ringR$), further it also distorts the model as $\calT_i$'s are no longer computable by simple $\PSU{}{}$ circuits. However, with careful analytically analysis we establish that the non-zeroness is preserved in the reduced model. The process is then repeated until we reach $k=1$, while maintaining the invariants which help us in preserving the non-zeroness till the end. We finish the proof by showing that the identity testing of reduced model can be done using known PIT algorithms.

\subsubsection{Jacobian hits again} In \autoref{thm:thm2} we exploit the prowess of the Jacobian polynomial first introduced in \cite{beecken2013algebraic} and later explored in \cite{agrawal2016jacobian} to unify known PIT algorithms and design new ones. Suppose we want to test $\sum_{i \in [k]} T_i \stackrel{?}{=} 0$, where $T_i \in \PSP{}{[\delta]}$ (respec.~$\PSU{}{}$). We associate the Jacobian $J(T_1, \hdots, T_r)$ to captures the algebraic independence of $T_1, \hdots, T_r$ assuming this to be a transcendence basis of the $T_i$’s. We design a variable reducing linear map $\Phi$ which preserves the algebraic independece of $T_1, \dots, T_r$ and show that for any $C$: $C(T_1,\hdots, T_k) = 0 \iff C(\Phi(T_1), \hdots, \Phi(T_k)) = 0$. Such a map is called `faithful' \cite{agrawal2016jacobian}. The map $\Phi$ ultimately provides a hitting set for $T_1 + \hdots +T_k$ , as we reduce to a PIT of a polynomial over `few' (roughly equal to $k$) variables, yielding a $\QP$-time algorithm.
\section{Preliminaries}  \label{sec:prel}
Before proving the results, we describe some of the assumptions and notations used throughout the paper. $\x$ denotes $(x_1,\hdots,x_n)$. $[n]$ denotes $\{1,\hdots,n\}$.

\subsection{Notations and Definitions}

\smallskip
\begin{itemize}
    \item \textbf{Logarithmic derivative.} Over a ring $\ringR$ and a variable $y$, the logarithmic derivative $\dlg_{y} : \ringR(y) \to \ringR(y)$ is defined as $\dlg_{y}(f):=\partial_y\,f/f$; here $\partial_{y}$ denotes the partial derivative with respect to variable $y$. One important property of $\dlg$ is that it is additive over a product as \[\dlg_{y}(f \cdot g) \,=\, \frac{\partial_y(f \cdot g)}{f \cdot g} \,=\, \frac{(f \cdot \partial_y g \,+\, g \cdot \partial_y f)}{f \cdot g}\,=\,\dlg_{y}(f) + \dlg_y(g).\] 
    \noindent We refer this effect as \emph{linearization} of product.

    \item \textbf{Circuit size.} Sparsity $\sp(\cdot)$ refers to the number of nonzero monomials. In this paper, it is a parameter of the circuit size. In particular, $\size(g_1 \cdots g_s)=\sum_{i \in [s]}\,(\sp(g_i)+\deg(g_i))$, for $g_i \in \Sigma \wedge$ (respectively~$\Sigma \Pi^{[\delta]}$). In whitebox settings, we also include the \emph{bit-complexity} of the circuit (i.e.~bit complexity of the constants used in the wires) in the size parameter. 
    Some of the complexity parameters of a circuit are \emph{depth} (number of layers), \emph{syntactic degree} (the maximum degree polynomial computed
    by any node), {\em fanin} (maximum number of inputs to a node).

    \item \textbf{Hitting set.} A set of points $\calH \subseteq \F^n$ is called a \emph{hitting-set} for a class $\calC$ of $n$-variate polynomials if for any nonzero polynomial $f \in \calC$, there exists a point in $\calH$ where $f$ evaluates to a nonzero value. A $T(n)$-time hitting-set would mean that the hitting-set can be generated in time $T(n)$, for input size $n$. 

    \item \textbf{Valuation.} Valuation is a map $\v_y : \ringR[y] \to \Z_{\ge 0}$, over a ring $\ringR$, such that $\v_y(\cdot)$ is defined to be the maximum power of $y$ dividing the element. It can be easily extended to fraction field $\ringR(y)$, by defining $\v_{y}(p/q):= \v_y(p) - \v_{y}(q)$; where it can be negative.
    
    \item \textbf{Field.} We denote the underlying field as $\F$ and assume that it is of characteristic $0$. All our results hold for other fields (eg.~$\Q_p, \F_p$) of \emph{large} characteristic (see Remarks in Section \ref{sec:pf:thm1}-\ref{sec:pf-thm2}). 
    
    \item \textbf{Jacobian.} The Jacobian of a set of polynomials ${\mathbf f}=\{f_1,\hdots,f_m\}$ in $\F[\x]$ is defined to be the matrix $\calJ_{\x}({\mathbf f}) := \left(\partial_{x_j}(f_i)\right)_{m \times n}$. Let $S \subseteq \x = \{x_1,\hdots,x_n\}$ and $|S|=m$. Then, polynomial $J_{S}({\bf f})$ denotes the minor (i.e.~determinant
    of the submatrix) of $\calJ_{\x}({\bf f})$, formed by the columns corresponding to the variables in $S$.
\end{itemize}

\subsection{Basics of Algebraic Complexity Theory} \label{sec:basic-act}

For detailed discussion on the basics of Algebraic Complexity Theory we will encourage readers to refer \cite{shpilka2010arithmetic,saxena2009progress,Mah13,saxena2014progress,saptharishi2019survey}. Here we will formally state a few of the PIT results and properties of circuits for the later reference.

\subsubsection*{Trivial PIT Algorithm}

The simplest PIT algorithm for any circuit in general is due to Polynomial Identity Lemma \cite{ore1922hohere,DL78,Zip79,schwartz1980fast}. When the number of variables is small, say $O(1)$, then this algorithm is very efficient.

\begin{lemma}[Trivial PIT] \label{lem:trivial-pit}
    For a class of $n$-variate, individual degree $ < d$ polynomial $f \in \F[\x]$ there exists a deterministic PIT algorithm which runs in time $O(d^n)$.
\end{lemma}

\subsubsection*{Sparse Polynomial}
Sparse PIT is testing the identity of polynomials with bounded number of monomials. There have been a lot of work on sparse-PIT, interested readers can refer \cite{ben1988deterministic,klivans2001randomness} and references therein. For the proof of poly-time hitting set of Sparse PIT see \cite[Thm.~2.1]{saxena2009progress}.

\begin{theorem}[Sparse-PIT map \cite{klivans2001randomness}] \label{thm:sparse-ks}
Let $p(\x) \in \F[\x]$ with individual degree at most $d$ and sparsity at most $m$. Then, there exists $1 \le r \le (mn \log d)^2$, such that \[ p(y,y^d,\hdots,y^{d^{n-1}}) \ne 0,\bmod\,y^r-1.\] If $p$ is computable by a size-$s$ $\Sigma \Pi$ circuit, then there is a deterministic algorithm to test its identity which runs in time $\poly(s,m)$.
\end{theorem}

Indeed if identity of sparse polynomial can be tested efficiently, product of sparse polynomial can be tested efficiently. We formalise this in the following:

\begin{lemma}[\cite{Sap13} Lemma 2.3] \label{lem:prod-sparse-PIT}
    For a class of $n$-variate, degree $d$ polynomial $f \in \F[\x]$ computable by $\PSP{}{}$ of size $s$, there is a deterministic PIT algorithm which runs in time $\poly(s,d)$.
\end{lemma}

A set $\calH \subseteq \F^n$ is called a Hitting Set for a class polynomial $\calC \subseteq \F[\x]$, if for all $g \in \calC$ \[ g \neq 0 \iff \exists \a \in \calH : g(\a) \neq 0.\]
In literature, PIT has a close association with Hitting set as the two notions are provably equivalent (refer Lemma 3.2.9 and 3.2.10 \cite{Forbes14Thesis}). Note that the set $\calH$ works for every polynomial of the class. Instead of a PIT algorithm occasionally we will use such a set.

\begin{lemma}[Hitting Set of $\PSU{}{}$] \label{lem:prod-sparse-hs}
    For a class of $n$-variate, degree $d$ polynomial $f \in \F[\x]$ computable by $\PSP{}{}$ of size $s$, there is an explicit Hitting Set of size $\poly(s,d)$.
\end{lemma}

\subsubsection*{Algebraic Branching Program (ABP)}
An ABP is a layered directed acyclic graph with $q+1$ many layers of vertices $V_0,\hdots,V_q$ with a source $a$ and a sink $b$ such that all the edges in the graph only go from $a$ to $V_0$, $V_{i-1}$ to $V_i$ for any $i\in [q]$, and $V_q$ to $b$. The edges have {\em uni}variate polynomials as their weights. 
The ABP is said to compute the polynomial \[ f(\x)\;=\;\sum_{p \in \mathsf{paths}(a,b)}\, \prod_{e \in p}\,W(e)\;,\]
where $W(e)$ is the weight of the edge $e$. The ABP has width-$w$ if $|V_i| \le w$, $\forall i \in \{0,\ldots,q\}$. In an equivalent definition, polynomials computed by ABP are of the form $A^T (\prod_{i \in [q]}\, D_i) B$, where $A, B \in \F^{w \times 1}[\x]$, and $D_i \in \F^{w \times w}[\x]$, where entries are univariate polynomials. We encourage interested readers to refer \cite{shpilka2010arithmetic, Mah13} for more detailed discussion.

\begin{definition}[Read-once oblivious ABP (ROABP)]
An ABP is called a {\em read-once oblivious ABP (ROABP)} if the edge weights are univariate polynomials in \emph{distinct} variables across layers. Formally, there is a permutation $\pi$ on the set $[q]$ such that the entries in the $i$-th matrix $D_i$ are univariate polynomials over the variable $x_{\pi(i)}$, i.e.,~they come from the polynomial
ring $\F[x_{\pi(i)}]$.
\end{definition}

A polynomial $f(x)$ is said to be computed by width-$w$ ROABPs in \emph{any order}, if for every permutation $\sigma$ of the variables, there exists a width-$w$ ROABP in the variable order $\sigma$ that computes the polynomial $f(\x)$. In whitebox setting, identity testing of any-order ROABP is completely solved. 

\begin{theorem}[Theorem 2.4 \cite{raz2005deterministic}] \label{thm:white-pit-ro}
    For $n$-variate polynomials computed by size-$s$ ROABP, a hitting set of size $O(s^5 + s \cdot n^4)$ can be constructed.    
\end{theorem}

There have been quite a few results on blackbox PIT for ROABPs as well \cite{forbes2013quasipolynomial,forbes2014hitting,gurjarARO}. The current best known algorithm works in quasipolynomial time.  

\begin{theorem}[Theorem 4.9 \cite{gurjarARO}] \label{thm:roabp-blackbox-pit}
For $n$-variate, individual-degree-$d$ polynomials computed by width-$w$ ROABPs in any order, a hitting set of size $(ndw)^{O(\log \log w)}$ can be constructed.
\end{theorem}

\subsubsection*{Depth-4 Circuits}

A polynomial $f(\x) \in \F[\x]$ is computable by $\SWSP{}{}{[\delta]}$ circuits if $ f(\x) = \sum_{i \in [s]} f_i(\x)^{e_i}$ where $\deg f_i \leq \delta$. The first nontrivial PIT algorithm for this model was designed in \cite{forbes2015deterministic}.

\begin{theorem}[Proposition 4.18 \cite{forbes2015deterministic}] \label{thm:swsp-hs}
There is a $\poly(n,d, \delta \log s)$-explicit hitting set of size $(nd)^{O(\delta \log s)}$ for the class of $n$-variate, degree-($\le d$) polynomials $f(\x)$, computed by $\SWSP{}{}{[\delta]}$-circuit of size $s$.
\end{theorem}

Similarly, $\SWSU{}{}{}$ circuits compute polynomials of the form $f(\x) = \sum_{ i \in [s] } f_i^{e_i}$ where $f_i$ is a sum of univariate polynomials. Using duality trick \cite{saxena2008diagonal} and PIT results from \cite{raz2005deterministic,gurjarARO}, one can design efficient PIT algorithm for $\SWSU{}{}{}$ circuits.

\begin{lemma} [PIT for $\SWSU{}{}$-circuits]\label{lem:swsu-hs}
Let  $P \in \SWSU{}{}$ of size $s$. Then, there exists a $\poly(s)$ (respectively~$s^{O(\log \log s)}$) time whitebox (respectively~blackbox) PIT for the same.
\end{lemma}

\begin{proof}[Proof sketch]
We show that any $g(\x)^e = (g_1(x_1)+\hdots +g_n(x_n))^e$, where $\d(g_i) \le s$ can be written as $ \sum_{j} h_{j1}(x_1) \cdots h_{jn}(x_n)$, for some $h_{j\ell} \in \F[x_{\ell}]$ of degree at most $es$. Define, $G := (y+g_1) \cdots (y+g_n)-y^n$. In its $e$-th power, notice that the leading-coefficient is $\cf_{y^{e(n-1)}}(G^e)= g^e$. So, interpolate on $e(n-1)+1$ many points ($y= \beta_i\in\F$) to get 
$$ \cf_{y^{e(n-1)}}(G^e) \;=\; \sum_{i=1}^{e(n-1)+1}\, \alpha_i\, G^e(\beta_i)\;.$$
Now, expand $G^e(\beta_i) = ((\beta_i +g_1) \cdots (\beta_i+g_n)-\beta_i^n)^e$, by binomial expansion (without expanding the inner $n$-fold product). The top-fanin can be atmost $s \cdot (e+1) \cdot (e(n-1)+1)= O(se^2n)$. The individual degrees of the intermediate univariates can be at most $es$. Thus, it can be computed by an ROABP (of {\em any order}) of size at most $O(s^2e^3n)$. 

Now, if $f=\sum_{j \in [s]} f_j^{e_j}$ is computed by a $\SWSU{}{}{}$ circuit of size $s$, then clearly, $f$ can also be computed by an ROABP (of any order) of size at most $O(s^6)$. So, the whitebox PIT follows from \autoref{thm:white-pit-ro}, while the blackbox PIT follows from Theorem \autoref{thm:roabp-blackbox-pit}.
\end{proof}

Further, $\SWSU{}{}$ can be shown to be closed under multiplication i.e., product of two polynomials, each computable by a $\SWSU{}{}$ circuit, is computable by a single $\SWSU{}{}$ circuit. To prove that we will need an efficient way to write a product of a few powers as a sum of powers, using simple interpolation. For an algebraic proof, see  \cite[Proposition~4.3]{CARLINI20125}. 

\begin{lemma}[Waring Identity for a monomial] \label{lem:waring-rank}
Let $M= x_1^{b_1} \cdots x_k^{b_k}$, where $1 \le b_1 \le \hdots \le b_k$, and {\em roots of unity} $\mathcal{Z}(i):=\{z \in \C : z^{b_i+1}=1\}$. Then, $$M= \sum_{\varepsilon(i) \in \mathcal{Z}(i): i=2,\cdots,k}\, \gamma_{\varepsilon(2),\hdots,\varepsilon(k)} \cdot \left(x_1 + \varepsilon(2)x_2 + \hdots + \varepsilon(k)x_k\right)^d\;,$$ where  $d:= \deg(M)=b_1+\hdots+b_k$, and  $\gamma_{\varepsilon(2),\hdots,\varepsilon(k)}$ are scalars ($\rk(M):=\prod_{i=2}^k\,(b_i+1)$ many).       
\end{lemma}

{\noindent \em Remark.} We actually need not work with $\F=\C$. We can go to a small extension (at most $d^k$), for a monomial of degree $d$, to make sure that $\varepsilon(i)$ exists. 

Using the above lemma we prove the closure result.

\begin{lemma}\label{lem:prod-swsw}
Let $f_i(\x,y) \in \F[y][\x]$, of syntactic degree $\le d_i$, be computed by a $\SWSU{}{}$ circuit of size $s_i$, for $i \in [k]$ (wrt $\x$). Then, $f_1 \cdots f_k$ has $\SWSU{}{}$ circuit of size $O((d_2+1) \cdots (d_k+1)\cdot\,s_1\cdots s_k)$.
\end{lemma}

\begin{proof}
Let $f_i=\sum_{j} f_{ij}^{e_{ij}}$; by assumption $e_{ij} \le d_i$ (by assumption). Then using \autoref{lem:waring-rank}, $f_{1j_1}^{e_{1j_1}} \cdots f_{kj_k}^{e_{kj_k}}$ has size at most $(d_2+1) \cdots (d_k+1) \cdot \left(\sum_{i \in [k]}\,\size(f_{ij_i})\right)$, for indices $j_1,\hdots, j_k$. Summing up for all $s_1 \cdots s_k$ many products (atmost) gives the upper bound.
\end{proof}
\section{Whitebox PIT for \texorpdfstring{$\SPSU{[k]}{}{}$}{SkPSU}} \label{sec:pf:thm1}

We consider a bloated model of computation which naturally generalizes $\SPSU{}{}{}$ circuits and works ideally under the $\DiDI$-techniques.

\begin{definition}
We call a circuit $\calC \in \Gen{k}{s}$, over $\ringR(\x)$, for any ring $\ringR$, with parameter $k$ and size-$s$, if $\calC \in \Sigma^{[k]} (\PSU{}{}/\PSU{}{}) \cdot (\SWSU{}{}/\SWSU{}{})$. It computes $f \in \ringR(\x)$, if $f\;=\; \sum_{i=1}^k\, T_{i}$, where 
    \begin{itemize}
        \item $T_i\,=:\,(U_i/V_i)\cdot (P_i/Q_i)$, for $U_i, V_i \in \PSU{}{}$, and $P_i,Q_i \in \SWSU{}{}$,  
        \item $\size(T_i) = \size(U_i) + \size(V_i) + \size(P_i) + \size(Q_i)$, and $\size(f)=\sum_{i \in [k]} \size(T_i)$.
    \end{itemize}
\end{definition}

\noindent It is easy to see that all size-$s$ $\SPSU{[k]}{}{}$ circuit are in $\Gen{k}{s}$. We will design the \emph{recursive} algorithm on $\Gen{k}{s}$. 

\begin{proof}[Proof of \autoref{thm:thm1}] 
Begin with defining $T_{i,0}:=T_i$ and $f_{0}:=f$ where $T_{i,0} \in \PSU{}{}$;  $\sum_{i} T_{i,0}=f_0$, and $f_0$ has size $\le s$. Assume $\d(f)<d\le s$; we keep the parameter $d$ separately, to help optimize the complexity later. In every recursive call we work with $\Gen{\cdot}{\cdot}$ circuits. 

As the input case, define $U_{i,0}:=T_{i,0}$ and $V_{i,0}:=P_{i,0}:= Q_{i,0}:=1$. We will use the hitting set of product of sparse polynomials (refer section \ref{sec:basic-act}) to obtain a point $\a = (a_1, \dots, a_n) \in \F^n$ such that $U_{i,0} \rvert_{\x = \a} \neq 0$, for all $i \in [k]$. Eventually this evaluation point will help in maintaining the invertibility of $\PSU{}{}$. Consider
\begin{align*}
    g \;:=\; \prod_{i \in [k]} T_{i,0} &\;=\; \prod_{i \in [k]} U_{i,0} \;=\; \prod_{i \in [\ell]} \sum_{j \in [n]} f_{ij}(x_j)\,,
\end{align*}
where $f_{ij}(x_j)$ are univariate polynomials of degree at most $d$ and $\ell \leq k \cdot s$. Note that $\deg g \leq d \cdot k \cdot s$ and $g$ is computable by a $\PSU{}{}$ circuit of size $O(s)$. Invoke \autoref{lem:prod-sparse-hs} to obtain a hitting set $\calH$, then evaluate $g$ on every point of $\calH$ to find an element $\a \in \calH$ such that $g(\a) \neq 0$. We emphasise that in whitebox setting all $U_{i,0}$, are readily available for evaluation. Since, the size of the set is $\poly(s)$ and each evaluation takes $\poly(s)$ time, this preliminary step will add $\poly(s)$ time to the overall time complexity. Moreover, we obtain the $\a \in \F^n$ which possess the required property. 

To capture the non-zeroness, consider a 1-1 homomorphism $\Phi : \F[\x] \longrightarrow \F[\x,z]$ such that $x_i \mapsto z\cdot x_i + a_i$ where $a_i$ is the $i$-th coordinate of $\a$, obtained earlier. Invertibility implies that $f_{0} = 0 \iff \Phi(f_{0}) = 0$. 
Now we proceed with the recursive algorithm which first reduces the identity testing from top-fanin $k$ to $k-1$. Note: $k=1$ is trivial.

\subsubsection*{First Step: Efficient reduction from $k$ to $k-1$}
By assumption, $\sum_{i=1}^k\,T_{i,0}\,=\,f_{0}$ and $T_{k,0}\ne0$. Apply $\Phi$ both sides, then divide and derive:
\begin{align} \label{Eq:thm1-k-k-1}
    \sum_{i \in [k]}\,T_{i,0} \;=\; f_0 \; &\iff\;  \sum_{i \in [k]}\,\Phi(T_{i,0}) \;=\; \Phi(f_0) \nonumber\\ 
    &\iff  \sum_{i \in [k-1]}\, \frac{\Phi(T_{i,0})}{\Phi(T_{k,0})}\,+\, 1 \;=\; \frac{\Phi(f_0)}{\Phi(T_{k,0})} \nonumber\\ 
    &\implies \sum_{i \in [k-1]}\,\partial_{z} \left( \frac{\Phi(T_{i,0})}{\Phi(T_{k,0})}\right) \;=\; \partial_{z}\left(\frac{\Phi(f_0)}{\Phi(T_{k,0})}\right) \nonumber\\ 
&\iff \sum_{i=1}^{k-1}\,\frac{\Phi(T_{i,0})}{\Phi(T_{k,0})} \cdot \dlg_{z} \left( \frac{\Phi(T_{i,0})}{\Phi(T_{k,0})} \right) \;=\; \partial_{z} \left( \frac{\Phi(f_0)}{\Phi(T_{k,0})}\right)\;.
\end{align}
Here onwards we say \(\dlg\) to mean \(\dlg_z\), unless stated otherwise. Define the following:
\begin{itemize}
    \setlength\itemsep{1em}
    \item $\ringR_1\,:=\, \F[z]/\langle z^{d} \rangle$. Note that, \autoref{Eq:thm1-k-k-1} holds over $\ringR_1(\x)$.
    \item $\widetilde{T}_{i,1}:=\Phi(T_{i,0})/\Phi(T_{k,0}) \cdot \dlg (\Phi(T_{i,0})/\Phi(T_{k,0}))$, $\forall\;i \in [k-1]$.
    \item $ f_{1}:=\partial_{z}(\Phi(f_0)/\Phi(T_{k,0}))$, over $\ringR_1(\x)$.
\end{itemize}
\vspace*{1em}

\textbf{Definability of $T_{i,1}$ and $f_1$.} It is easy to see that these are well-defined terms. Here, we emphasize that we do not exactly compute/store $\widetilde{T}_{i,1}$ as a fraction where the degree in $z$ is $<d$; instead it is computed as an element in $\F(z,\x)$, where $z$ is a formal variable. Formally, we compute $T_{i,1} \in \F(z,\x)$, such that $\widetilde{T}_{i,1} = T_{i,1}$, over $\ringR_1(\x)$. We keep track of the degree of $z$ in $T_{i,1}$. 
Thus, $\sum_{i \in [k-1]}\,T_{i,1}=f_1$, over $\ringR_1(\x)$.

\vspace*{1em}
\textbf{The `iff' condition.} To show that our one step of $\DiDI$ has reduced to the identity testing of $\Gen{k-1}{\cdot}$, we need an $\iff$ condition. So far equality in \autoref{Eq:thm1-k-k-1} over $\ringR_1(\x)$ is \emph{one-sided}. Note that $f_{1} \ne 0$ implies $\v_{z}(f_{1}) < d=:d_1$. By assumption, $\Phi(T_{k,0})$ is invertible over $\ringR_1(\x)$. Further, $f_{1} = 0$, over $\ringR_1(\x)$, which implies --

\begin{enumerate}
    \setlength\itemsep{1em}
    \item Either, $\Phi(f_0)/\Phi(T_{k,0})$ is $z$-free. Then $\Phi(f_0)/\Phi(T_{k,0}) \,\in \F(\x)$, which further implies it is in $\F$, because of the map $\Phi$ ($z$-free implies $\x$-free, by substituting $z=0$). Also, note that $f_0, T_{k,0} \ne 0$ implies $\Phi(f_0)/\Phi(T_{k,0})$ is a \emph{nonzero} element in $\F$. Thus, it suffices to check whether  $\Phi(f_0)\rvert_{z=0}$ is non-zero or not. 
    \item Or, $\partial_{z}(\Phi(f_0)/\Phi(T_{k,0}))= z^{d_1} \cdot p$ where $p \in \F(z,\x)$ s.t.~$\v_{z}(p) \ge 0$. By simple power series expansion, one can show that $p \in \F(x)[[z]]$. 
    
    \begin{lemma}[Valuation]\label{lem:val-powerseries}
        Consider $f \in \F(\x,y)$ such that $\v_{y}(f) \ge 0$. Then, $f \in \F(\x)[[y]]\,\bigcap\,\F(\x,y)$.
    \end{lemma}
    \begin{proof-sketch}
        Let $f=g/h$, where $g,h \in \F[\x,y]$. Now, $\v_{y}(f) \ge 0$, implies $\v_y(g) \ge \v_y(h)$. Let $\v_{y}(g)=d_1$ and $\v_{y}(h)=d_2$, where $d_1 \ge d_2 \ge 0$. Write $g=y^{d_1}\cdot \tilde{g}$ and $h=y^{d_2}\cdot \tilde{h}$. Write, $\tilde{h}= h_0+h_1\,y+h_2\,y^2+\hdots + h_d\,y^d$, for some $d$. Note that $h_0 \ne 0$. Thus, 
        \begin{align*}
            f&\;=\; y^{d_1-d_2} \cdot \tilde{g}/(h_0+h_1\,y+\hdots+h_d\,y^d) \\ &\;=\;y^{d_1-d_2} \cdot (\tilde{g}/h_0) \cdot (1+(h_1/h_0)\,y+\hdots + (h_d/h_0)\,y^d)^{-1}\; \in \F(\x)[[y]]\;.
        \end{align*}
        The last conclusion follows by the inverse identity in the power-series ring.
    \end{proof-sketch}

    Hence, $\Phi(f_0)/\Phi(T_{k,0}) = z^{d_1+1} \cdot q$ where $q \in F(\x)[[z]]$, i.e. \[\Phi(f_0)/\Phi(T_{k,0}) \in \langle z^{d_1+1}\rangle_{\F(\x)[[z]]} \implies \v_{z}(\Phi(f_0)) \ge d+1,\] a contradiction.
\end{enumerate}

\noindent Conversely, it is obvious that $f_0=0$ implies $f_1=0$. Thus, we have proved the following 
\[ \sum_{i \in [k]}\,T_{i,0}\,\ne 0 \;\text{ over}\;\F[\x] \iff \sum_{i \in [k-1]}\, T_{i,1} \ne 0 \;\text{ over} \;\ringR_1(\x),\;\text{ or}\,,\;\;0 \ne \Phi(f_0) \rvert_{z=0} \in \F\;.\]
 Eventually, we show that $T_{i,1} \in (\PSU{}{}/\PSU{}{}) \cdot (\SWSU{}{}/\SWSU{}{})$, over $\ringR_1(\x)$, with polynomial blowup in size (\autoref{cl:size-bound-on-Ti}). So, the above circuit is in $\Gen{k-1}{\cdot}$, over $\ringR_1(\x)$, which we recurse on to finally give the identity testing. The subsequent steps will be a bit more tricky:

\vspace*{1em}
\subsubsection*{Induction step} 

Assume that we are in the $j$-th step ($j\ge1$). Our induction hypothesis assumes --

\begin{enumerate}
    \item $\sum_{i \in [k-j]}\, T_{i,j} = f_j$, over $\ringR_j(\x)$, where $\ringR_j:= \F[z]/\langle z^{d_j}\rangle$ for \(d_j < d\), and $T_{i,j} \ne 0$. \label{hypo1}
    \item $\v_{z}(T_{i,j}) \ge 0, \forall i \in [k-j]$. \label{hypo2}
    \item Non-zero preserving iff condition
    \begin{align*} 
        f \neq 0 \text{, over } \F[\x] \iff & f_j \ne 0, \text{ over } \ringR_j(\x), \\
        & \text{or }\bigvee_{i=0}^{j-1} \left((f_i/T_{k-i,i}) \rvert_{z=0}\ne 0,\, \text{over}\,\F(\x)\right)
    \end{align*} \label{hypo3}
    \item Here, $T_{i,j}=: \left( U_{i,j}/V_{i,j} \right) \cdot \left( P_{i,j}/Q_{i,j} \right)$, where $U_{i,j}, V_{i,j} \in \PSU{}{}$, and $P_{i,j}, Q_{i,j} \in \SWSU{}{}$, each in $\ringR_j[\x]$. Think of them being computed as $\F(z,\x)$, with the degrees being tracked. Wlog, assume that $\v_{z}(T_{k-j,j})$ is the minimal among all $T_{i,j}$'s. \label{hypo4}
    \item $U_{i,j}\rvert_{z=0}, V_{i,j}\rvert_{z=0}  \in \F\backslash \{0\}$. \label{hypo5}
\end{enumerate} 

We follow as before without applying homomorphism any further. Note that the `or condition' in the hypothesis \ref{hypo3} is similar to the $j=0$ case except that there is no $\Phi$: this is because $\Phi(f_0)\rvert_{z=0} \ne 0 \iff \Phi(f_0/T_{k,0})\rvert_{z=0} \ne 0$. This condition just separates the derivative from the constant-term.

\textbf{Efficient reduction from $k-j$ to $k-j-1$.} Let $\v_{z}(T_{i,j})=:v_{i,j}$, for all $i \in [k-j]$. Note that \[\min_{i} \v_{z}(T_{i,j})=\min_{i}\v_{z}(P_{i,j}/Q_{i,j})=v_{k-j,j}\] since $\v_{z}(U_{i,j}) =\v_{z}(V_{i,j}) =0$ (else we reorder). We remark that  $ 0 \le v_{i,j} < d_j$ for all $i$'s in $j$-th step; upper-bound is strict, since otherwise $T_{i,j}=0$ over $\ringR_j(x)$.

Similar to the first step, we divide with $T_{k-j,j}$ which has $\min \v$ and then derive:
\begin{align}
    \sum_{i \in [k-j]}\,T_{i,j} \;=\; f_j  &\iff  \sum_{i \in [k-j-1]}\,T_{i,j}/T_{k-j,j} \,+\, 1 \;=\; f_j/T_{k-j,j} \nonumber\\  &\implies \sum_{i \in [k-j-1]}\,\partial_{z} (T_{i,j}/T_{k-j,j}) \;=\; \partial_{z}(f_j/T_{k-j,j})\nonumber \\ &\iff \sum_{i=1}^{k-j-1}\,T_{i,j}/T_{k-j,j} \cdot \dlg (T_{i,j}/T_{k-j,j}) \;=\; \partial_{z} (f_j/T_{k-j,j}) \label{eq:induction-thm1-main-eq}
\end{align}
Define the following:
\begin{itemize}
    \setlength\itemsep{1em}
    \item $\ringR_{j+1}:= \F[z]/ \langle z^{d_{j+1}} \rangle$, where $d_{j+1}:=d_j-v_{k-j,j}-1$. 
    \item $\widetilde{T}_{i,j+1} := T_{i,j}/T_{k-j,j} \cdot \dlg(T_{i,j}/T_{k-j,j})$, $\forall\;i \in [k-j-1]$.
    \item $f_{j+1}:=\partial_{z}(f_j/T_{k-j,j})$, over $\ringR_{j+1}(\x)$.
\end{itemize}

We emphasize on the fact again that we do not exactly compute $\widetilde{T}_{i,j+1}$ mod $z^{d_{j+1}}$; instead it is computed as a fraction in $\F(z,\x)$, with formal $z$. Formally, we compute $T_{i,j+1} \in \F(z,\x)$, such that $\widetilde{T}_{i,j+1} = T_{i,j+1}$, over $\ringR_{j+1}(\x)$. We keep track of the degree of $z$ in $T_{i,j+1}$. Next, we will show that all the inductive hypotheses assumed hold in the $j^{\text{th}}$ step as well. 

\smallskip

\noindent \textbf{Hypothesis (\ref{hypo1}): Definability of $T_{i,j+1}$ and $f_{j+1}$.} By the minimal valuation assumption, it follows that $\v(f_j) \ge v_{k-j,j}$, and thus $\widetilde{T}_{i,j+1}$ and $f_{j+1}$ are all well-defined over $\ringR_{j+1}(\x)$. Note that, \autoref{eq:induction-thm1-main-eq} holds over $\ringR_{j+1}(\x)$ as $d_{j+1} < d_j$ (because, whatever identity holds true $\bmod \ z^{d_j}$ must hold $\bmod \ z^{d_{j+1}}$ as well). Hence, we must have $\sum_{i=1}^{k-j-1}\, \widetilde{T}_{i,j+1} = f_{j+1}$, over $\ringR_{j+1}(\x)$ thus proving the induction hypothesis (\ref{hypo1}). 
 

\smallskip

\noindent \textbf{Hypothesis (\ref{hypo2}): Positivity of Valuation.} Since we divide by the $\min \v$, by definition we immediately get $\v_{z}(T_{i,j+1}) \ge 0$ proving the hypothesis. Further, we claim that min $\v$ computation in $\DiDI$ is easy. For this, recall from the definition of valuation \[\min_{i}\v_{z}(P_{i,j}/Q_{i,j})= \min_{i} (\v_{z}(P_{i,j}) - \v_{z}(P_{i,j}) ). \] Therefore, for $\min \v$ we compute $\v_{z}(P_{i,j})$ and $\v_{z}(Q_{i,j})$ for all $i \in [k-j]$.

Here is an important lemma which shows that coefficient of $y^e$ of a polynomial $f(\x,y) \in \F[\x,y]$, computed by a $\SWSU{}{}$ circuit, can be computed by a small $\SWSU{}{}$ circuit.

\begin{lemma}[Coefficient extraction]\label{lem:cf-swsw}
Let $f(\x,y) \in \F[y][\x]$ be computed by a $\SWSU{}{}$ circuit of size $s$ and degree $d$. Then, $\cf_{y^e}(f) \in \F[\x]$ can be computed by a small $\SWSU{}{}$ circuit of size $O(sd)$, over $\F[\x]$.
\end{lemma}
\begin{proof-sketch}
    Let, $f= \sum_{i} \alpha_i \cdot g_i^{e_i}$. Of course, $e_i \le s$ and $\d_{y}(f) \le d$. Thus, write $f=\sum_{i=0}^{d}\,f_i\cdot y^i$, where $f_i \in \F[\x]$. We can interpolate on $d+1$-many distinct points $y\in \F$ and conclude that $f_i$ has a $\SWSU{}{}{}$ circuit of size at most $O(sd)$. 
\end{proof-sketch}

\noindent Using \autoref{lem:cf-swsw} we known $\cf_{z^e}(P_{i,j})$ and $\cf_{z^e}(Q_{i,j})$ are in $\SWSU{}{}$ over $F[\x]$. We can keep track of $z$ degree and thus interpolate to find the minimum $e < d_j$ such that the computed coefficients are $\ne 0$, which gives the respective $\v$.

\smallskip

\noindent \textbf{Hypothesis (\ref{hypo3}): The `iff' condition.} The above \autoref{eq:induction-thm1-main-eq} pioneers to reduce from $k-j$-summands to $k-j-1$. But we want a $\iff$ condition to efficiently reduce the identity testing. If $f_{j+1} \ne 0$, then $\v_{z}(f_{j+1}) < d_{j+1}$.  Further, $f_{j+1} = 0$, over $\ringR_{j+1}(\x)$ implies--
\begin{enumerate}
    \item Either, $f_j/T_{k-j,j}$ is $z$-free. This implies it is in $\F(\x)$. Now, if indeed $f_0 \ne 0$, then the computed $T_{i,j}$ as well as $f_j$ must be non-zero over $\F(z,\x)$, by induction hypothesis (as they are non-zero over $\ringR_j(\x)$).  However,
    \begin{align*}
    \left(\frac{T_{i,j}}{T_{k-j,j}}\right) \bigg\rvert_{z=0} &= \left( \frac{U_{i,j} \cdot V_{k-j,j}}{U_{k-j,j} \cdot V_{i,j}} \right) \bigg\rvert_{z=0} \cdot \left( \frac{P_{i,j} \cdot Q_{k-j,j}}{P_{k-j,j} \cdot Q_{i,j}} \right) \bigg\rvert_{z=0} \\ & \in\; \F \cdot \left( \frac{\SWSU{}{}{}}{\SWSU{}{}{}} \right).
    \end{align*} Thus, \[ \frac{f_j}{T_{k-j,j}} \; \in \; \sum \;  \F \cdot \left( \frac{\SWSU{}{}{}}{\SWSU{}{}{}} \right) \,\in\,  \left( \frac{\SWSU{}{}{}}{\SWSU{}{}{}} \right).\] Here we crucially use that $\SWSU{}{}{}$ is closed under multiplication (\autoref{lem:prod-swsw}). 
    Thus, this identity testing can be done in $\poly$-time (\autoref{lem:swsu-hs}).  For, detailed time-complexity and calculations, see \autoref{cl:size-bound-on-Ti} and its subsequent paragraph. \smallskip 
    \item Or, $\partial_{z}(f_j/T_{k-j,j}) = z^{d_{j+1}} \cdot p$, where $p \in \F(z,\x)$ s.t.~$\v_{z}(p) \ge 0$. By a simple power series expansion, one concludes that $p \in \F(\x)[[z]]$ (\autoref{lem:val-powerseries}). Hence, one concludes that \[ \frac{f_j}{T_{k-j,j}} \in  \left\langle z^{d_{j+1}+1}\right\rangle_{\F(\x)[[z]]}\; \implies \v_{z}(f_j) \ge d_j,\] i.e.~$f_j = 0$, over $\ringR_j(\x)$.
\end{enumerate}

\medskip
Conversely, $f_j=0$, over $\ringR_j(\x)$, implies \begin{align*}\v_{z}(f_j) \ge d_j &\implies \v_{z} \left( \partial_{z}\left( \frac{f_j}{T_{k-j,j}} \right) \right) \ge d_j-v_{k-j,j}-1 \\ &\implies f_{j+1}=0, \; \text{over} \; \ringR_{j+1}(\x).\end{align*}
Thus, we have proved that $\sum_{i \in [k-j]}\,T_{i,j}\,\ne 0 \;\text{ over}\;\ringR_j(\x)$ iff \[\sum_{i \in [k-j-1]}\, T_{i,j+1} \ne 0 \;\text{ over} \;\ringR_{j+1}(\x)\;,\,\text{or}\,,\;\;0 \ne \left( \frac{f_j}{T_{k-j,j}} \right)\bigg\rvert_{z=0} \in \F(\x)\;.\]
Therefore induction hypothesis (\ref{hypo3}) holds. 


\smallskip
\noindent \textbf{Hypothesis (\ref{hypo4}): Size analysis.} We will show that $T_{i,j+1} \in (\PSU{}{}/\PSU{}{}) \cdot (\SWSU{}{}/\SWSU{}{})$, over $\ringR_{j+1}(\x)$, with only polynomial blowup in size.  Let $\size(T_{i,j}) \le s_j$, for $i \in [k-j]$, and $j \in [k]$. Note that, by assumption, $s_0 \le s$.

\begin{claim}[Final size] \label{cl:size-bound-on-Ti}
$T_{1,k-1} \in (\PSU{}{}/\PSU{}{}) \cdot (\SWSU{}{}/\SWSU{}{})$ of size $s^{O(k7^k)}$, over $\ringR_{k-1}(\x)$.
\end{claim}

\begin{proof}
Steps $j=0$ and $j>0$ are slightly different because of the $\Phi$. However the main idea of using power-series is the same which eventually shows that $\dlg(\Sigma \wedge) \in \SWSU{}{}{}$. 

We first deal with $j=0$. Let $A-z \cdot B= \Phi(g) \in \Sigma \wedge$, for some $A \in \F$ and $B \in \ringR_1[\x]$. Note that $A \ne 0$ because of the map $\Psi$. Further, $\size(B) \le O(d \cdot \size(g))$, as a single monomial of the form $x^e$ can produce $d+1$-many monomials. Over $\ringR_1(\x)$,
\begin{align}
    \dlg(\Phi(g)) = -\frac{\partial_{z}(B\cdot z)}{A(1- \frac{B}{A}\cdot z)}  = -\frac{\partial_{z}(B\cdot z)}{A} \cdot \sum_{i=0}^{d_1-1} \left(\frac{B}{A}\right)^{i} \cdot z^i\;. \label{powerseries-1/g-phi}
\end{align}
$B^i$ has a trivial $\WSU{}$-circuit of size $O(d \cdot \size(g))$.  Also, $\partial_{z}(B\cdot z)$ has a $\Sigma \wedge$-circuit of size at most $O(d\cdot \size(g))$. Using waring identity (\autoref{lem:waring-rank}), we get that each $\partial_{z}(B\cdot z) \cdot (B/A)^i \cdot z^i$ has size $O(i\cdot d \cdot \size(g))$, over $\ringR_1(\x)$. Summing over $i \in [d_1-1]$, the overall size is at most $ O(d_1^2\cdot d \cdot \size(g))=O(d^3 \cdot \size(g))$, as $d_0=d_1=d$. 

For the $j$-th step, we emphasize that the degree could be larger than $d$. Assume that syntactic degree of denominator and numerator of $T_{i,j}$ (each in $\F[\x,\z]$) are bounded by $D_{j}$ (it is {\em not} $d_j$ as seen above; this is to save on the trouble of mod-computation at each step). Of course, $D_0<d \le s$. 

For $j>0$, the above summation in \autoref{powerseries-1/g-phi} is over $\ringR_{j}(\x)$. However the degree could be $D_j$ (possibly more than $d_j$) of the corresponding $A$ and $B$. Thus, the overall size after the power-series expansion would be $O(D_j^2 \cdot d \cdot \size(g))$. 

Using \autoref{lem:diff-size}, we can show that $\dlg(P_{i,j}) \in \SWSU{}{}/\SWSU{}{}$ (similarly for $Q_{i,j}$), of size $O(D_{j}^2 \cdot s_j)$. Also $\dlg(U_{i,j} \cdot V_{k-j,j}) \in \sum \,\dlg(\Sigma \wedge)$, i.e. sum of action of $\dlg$ on $\Sigma \wedge$ (since $\dlg$ linearizes product); and it can be computed by the above formulation. 
Thus, $\dlg(T_{i,j}/T_{k-j,j})$ is a sum of $4$-many  $\SWSU{}{}/\SWSU{}{}$ of size at most $O(D_{j}^2\,s_j)$ and $1$-many $\SWSU{}{}$ of size $O(D_j^2 d_j s_j)$ (from the above power-series computation) [Note: we summed up the $\SWSU{}{}$-expressions from $\dlg(\Sigma \wedge)$ together]. Additionally the syntactic degree of each denominator and numerator (of the $\SWSU{}{}{}/\SWSU{}{}{}$) is $O(D_j)$. We rewrite the $4$ expressions (each of $\SWSU{}{}{}/\SWSU{}{}{}$) and express it as a single $\SWSU{}{}/\SWSU{}{}$ using waring identity (\autoref{lem:prod-swsw}), with the size blowup of $O(D_{j}^{12}\,s_j^{4})$; here the syntatic degree blowsup to $O(D_j)$. Finally we add the remaining $\SWSU{}{}$ circuit (of size $O(D_j^3 s_j)$ and degree $O(dD_j)$) to get $O(s_j^{5} D_{j}^{16}d)$. To bound this, we need to understand the degree bound $D_j$. 

Finally we need to multiply $T_{i,j}/T_{k-j,j} \in (\PSU{}{}/\PSU{}{}) \cdot (\SWSU{}{}/\SWSU{}{})$ where each $\SWSU{}{}$ is a product of two $\SWSU{}{}$ expression of size $s_j$ and syntactic degree $D_j$; clubbed together owing a blowup of $O(D_j \cdot s_j^2)$. Hence multiplying it with $\SWSU{}{} / \SWSU{}{}$ expression obtained from $\dlg$ computation above gives size blowup of $s_{j+1} = s^7 \cdot D_j^{O(1)} \cdot d$.

Computing $T_{i,j}/T_{k-j,j}$ increases the syntactic degree `slowly'; which is much less than the size blowup.
As mentioned before, the deg-blowup in $\dlg$-computation is $O(dD_j)$ and in the clearing of four expressions, it is just $O(D_j)$. Thus, $D_{j+1}= O(dD_j) \implies D_j= d^{O(j)}$. 
 
The recursion on the size is $s_{j+1} = s_j^{7} \cdot d^{O(j)}$. Using $d \le s$ we deduce, $s_j= (sd)^{O(j \cdot 7^j)}$. In particular, $s_{k-1}$, size after $k-1$ steps is $s^{O(k \cdot 7^k)}$. This computation quantitatively establishes induction hypothesis (\ref{hypo4}).
\label{deg-bound-thm1}
\end{proof}

\smallskip

\noindent \textbf{Hypothesis (\ref{hypo5}): Invertibility of $\PSU{}{}$-circuits.} For invertibility, we want to emphasise that the $\dlg$ compuation plays a crucial role here. In the following lemma we claim that the action $\dlg(\SWSU{}{}) \in \SWSU{}{}{}/\SWSU{}{}{}$, is of $\poly$-size. 

\begin{lemma}[Differentiation] \label{lem:diff-size}
    Let $f(\x,y) \in \F[y][\x]$ be computed by a $\SWSU{}{}$ circuit of size $s$ and degree $d$. Then, $\partial_{y}(f)$ can be computed by a small $\SWSU{}{}$ circuit of size $O(sd^2)$, over $\F[y][\x]$.
\end{lemma}
\begin{proof-sketch}
    \autoref{lem:cf-swsw} shows that each $f_e$ has $O(sd)$ size circuit where $f=\sum_{e} f_e\,y^e$. 
    Doing this for each $e \in [0,d]$ gives a blowup of $O(sd^2)$.
\end{proof-sketch}

Similarly consider the action on $ \PSU{}{}$. We know $\dlg$ distributes the product additively, so it suffices to work with $\dlg(\Sigma \wedge)$; and earlier in \autoref{cl:size-bound-on-Ti} we saw that $\dlg(\Sigma \wedge) \in \SWSU{}{}{}$ of $\poly$-size. Assuming these, we simplify \[ \frac{T_{i,j}}{T_{k-j,j}} = \frac{U_{i,j} \cdot V_{k-j,j}}{V_{i,j} \cdot U_{k-j,j}} \cdot \frac{P_{i,j} \cdot Q_{k-j,j}}{Q_{i,j} \cdot P_{k-j,j}},\] and its $\dlg$. Thus, using \autoref{eq:induction-thm1-main-eq}, $U_{i,(j+1)}$ grows to $U_{i,j} \cdot V_{k-j,j}$ (and similarly $V_{i,(j+1)}$). This also means: $U_{i,(j+1)} \rvert_{z=0}$ $\in \F\setminus\{0\}$ and thereby proving the hypothesis.

\subsubsection*{Final time complexity} 
The above proof actually shows that $T_{1,k-1}$ is in $\Gen{1}{s^{O(k\cdot 7^k)}}$ over $\ringR_{k-1}(\x)$; and that the degree bound on $z$ (over $\F[z,\x]$, keeping denominator and numerator `in place') is $D_{k-1} = d^{O(k)}$. We cannot directly use the identity testing algorithms of the constituent simpler models due to $\ringR_{k-1}(\x)$. Moreover, using hypothesis (\ref{hypo2}) and \autoref{lem:val-powerseries} we know that $T_{1,k-1} \in \F(\x)[[z]]$ and it suffices to do identity testing on the first term of the powerseries: $T_{1,k-1} \rvert_{z = 0}$ over $\F(\x)$. Note that, hypothesis (\ref{hypo5}) guarantees that $\PSU{}{}$ part remains non-zero on $z = 0$ evaluation, however, $\SWSU{}{}/\SWSU{}{}$ may be undefined. For this, we keep track of $z$ degree of numerator and denominator, which will be polynomially bounded as seen in the discussion above. We can easily interpolate and cancel the $z$ power to make it work. Basically this shows that to test $T_{1,k-1}$ we need to test $z^e \cdot \SWSU{}{}$ over $\F[\x]$ where $e \geq 0$ due to positive valuation. Whitebox PIT of $\SWSU{}{}$ is in poly-time using \autoref{lem:swsu-hs}, and testing $z^e$ is possible using \autoref{lem:trivial-pit} with appropriate degree bound. The proof above is constructive: we calculate $U_{i,j+1}$ (and other terms) from $U_{i,j}$ explicitly. Gluing everything together we conclude this part can be done in $s^{O(k 7^k)}$ time.

What remains is to test the $z=0$-part of induction hypothesis (\ref{hypo3}); it could \emph{short-circuit} the recursion much before $j=k-1$. As we mentioned before, in this case, we need to do a PIT on $\SWSU{}{}{}$ only. At the $j$-th step, when we substitute $z=0$, the size of each $T_{i,j}$ can be at most $s_j$ (by definition). We need to do PIT on a simpler model: $\sum^{[k-j]}\;\F \cdot (\SWSU{}{}{}/\SWSU{}{}{})$. We can clear out and express this as a single $\SWSU{}{}{}/\SWSU{}{}{}$ expression; with a size blowup of $s_j^{O(k-j)} \le (sd)^{O(j(k-j)7^j)}$. Since this case could short-circuit the recursion, to bound the final time complexity, we need to consider the $j$ which maximizes the exponent.

\begin{lemma}\label{max-exp-opt}
    Let $k \in \N$, and $h(x):= x (k-x)7^x$. Then, $\max_{i \in [k-1]} h(i)=h(k-1)$.
\end{lemma}
\begin{proof-sketch}
    Differentiate to get $h'(x)= (k-x)7^x - x7^x +  x(k-x)(\log 7)7^x = 7^x\cdot[x^2(-\log7) + x(k\log7-2) + k ]$. It vanishes at \[x = \left(\frac{k}{2} - \frac{1}{\log 7}\right) + \sqrt{\left(\frac{k}{2} - \frac{1}{\log 7}\right)^2 - \frac{k}{\log 7}}\;.\]Thus, $h$ is maximized at the integer $x=k-1$.
\end{proof-sketch}

Therefore, $\max_{j \in [k-1]} j(k-j)7^j = (k-1)7^{k-1}$. 
Finally, use \autoref{lem:swsu-hs} for the base-case whitebox PIT. Thus, the final time complexity is $s^{O(k \cdot 7^k)}$. 

Here we also remark that in $z =0$ substitution $\SWSU{}{}/\SWSU{}{}$ may be undefined. However, we keep track of $z$ degree of numerator and denominator, which will be polynomially bounded as seen in the discussion above. We can easily interpolate and cancel the $z$ power to make it work. 

\smallskip
{\bf \noindent Bit complexity.} It is routine to show that the  bit-complexity is really what we claim. Initially, the given circuit has bit-complexity $s$. The main blowup happens due to the $\dlg$-computation which is a poly-size blowup. We also remark that while using \autoref{lem:prod-swsw} (using \autoref{lem:waring-rank}), we \emph{may} need to go to a field extension of at most $s^{O(k)}$ (because of the $\varepsilon(i)$ and correspondingly the constants $\gamma_{\varepsilon(2),\hdots,\varepsilon(k)}$, but they still are $s^{O(k)}$-bits). Also, \autoref{thm:sparse-ks} and \autoref{lem:swsu-hs} computations blowup bit-complexity polynomially. This concludes the proof.
\end{proof}

\begin{remark*}
\begin{enumerate}
\setlength\itemsep{.1em}
    \item The above method does {\em not} give whitebox PIT (in $\poly$-time) for $\SPSP{[k]}{}{[\delta]}$, as we donot know $\poly$-time whitebox PIT for $\SWSP{}{}{[\delta]}$. However, the above methods do show that whitebox-PIT for $\SPSP{[k]}{}{[\delta]}$ polynomially \emph{reduces} to whitebox-PIT for $\SWSP{}{}{[\delta]}$.
    \item DiDI-technique can be used to give whitebox PIT for the general bloated model $\Gen{k}{s}$. 
    \item The above proof works when the characteristic is $\ge d$. This is because the nonzeroness remains \emph{preserved} after derivation wrt $z$.
\end{enumerate}
\end{remark*}

\subsection{Algorithm}

The whitebox PIT for \autoref{thm:thm1}, that is discussed in section \ref{sec:pf:thm1}, appears (below) as Algorithm \ref{algo:whitebox-algo}.

\begin{algorithm}
    \caption{Whitebox PIT Algorithm for $\SPSU{[k]}{}{}$-circuits \label{algo:whitebox-algo}}
    \hspace*{\algorithmicindent} \textbf{INPUT: $f=T_{1} + \hdots + T_k \in \SPSU{[k]}{}{}$, a whitebox circuit of size $s$ over $\F[\x]$} \\
    \hspace*{\algorithmicindent} \textbf{OUTPUT: $ 0$, if $f \equiv 0$, and $1$, if non-zero.} \\
    \begin{algorithmic}[1]
    \STATE{Let $\Psi : \F[\x] \longrightarrow \F[z]$, be a sparse-PIT map, using \cite{klivans2001randomness} (\autoref{thm:sparse-ks}). Apply it on $f$ and  check  whether $\Psi(f) \stackrel{?}{=} 0$. If non-zero, {\tt output} $1$}

    \STATE{Obtain a point $\a = (a_1, \dots, a_n) \in \F^n$ from Hitting Set $\calH$ of $\PSU{}{}$ such that $T_{i} \rvert_{\x = \a} \neq 0$, for all $i \in [k]$. And define $\Phi : x_i \mapsto z \cdot x_i + a_i$. Check $\sum_{i \in [k-1]} \partial_{z}(\Phi(T_i)/\Phi(T_k)) \stackrel{?}{=} 0 \bmod z^{d_1}$ ($d_1:=s$) as follows:}
    
    \STATE{Consider each $T_{i,1}:=\partial_{z}(\Phi(T_i)/\Phi(T_k))$ over $R_1(\x)$, where $R_1:=\F[z]/\langle z^{d_1} \rangle$. Use $\dlg$ computation (\autoref{cl:size-bound-on-Ti}), to write each $T_{i,1}$ in a `bloated' form as $(\PSU{}{}/\PSU{}{}) \cdot (\SWSU{}{}{}/\SWSU{}{}{})$.}
    
    \FOR{$j\leftarrow 1$ \TO $k-1$}
        \STATE{Reduce the top-fanin at each step using `{\bf Di}vide \& {\bf D}erive' technique. Assume that at $j$-th step, we have to check the identity:}
        $\sum_{i \in [k-j]} T_{i,j}\; \stackrel{?}{=}\;0\; \text{over}\;R_j(\x),\;\text{where}\;R_j:= \F[z]/\langle z^{d_j}\rangle\;$, 
        each $T_{i,j}$ has a $(\PSU{}{}/\PSU{}{}) \cdot (\SWSU{}{}{}/\SWSU{}{}{})$ representation and therein each $\PSU{}{} \rvert_{z=0} \in \F\setminus\{0\}$. 
        
        \STATE{Compute $v_{k-j,j}:= \min_i \v_{z}(T_{i,j})$; by reordering it is for $i=k-j$. 
        To compute $v_{k-j,j}$, use coefficient extraction (\autoref{lem:cf-swsw}) and $\SWSU{}{}{}$-circuit PIT (\autoref{lem:swsu-hs}).}
            
        \STATE{`{\bf D}ivide' by $T_{k-j,j}$ and check whether $\left(\sum_{i \in [k-j-1]}\; (T_{i,j} / T_{k-j,j}) +1 \right)\bigg\rvert_{z=0} \stackrel{?}{=} 0$. Note: this expression is in $(\SWSU{}{}{}/\SWSU{}{}{})$. Use--- (1) $\PSU{}{}\rvert_{z=0} \in \F$, and (2) \emph{closure} of $\SWSU{}{}{}$ under multiplication. Finally, do PIT on this by \autoref{lem:swsu-hs}.}
            
        \STATE{If it is non-zero, {\tt output} $1$, otherwise `{\bf D}erive' wrt $z$ and `{\bf I}nduct' on $\left(\sum_{i \in [k-j-1]} \partial_{z}(T_{i,j}/T_{k-j,j})\right) \stackrel{?}{=} 0$, over $R_{j+1}(\x)$ where $R_{j+1}:= \F[z]/\langle z^{d_j-v_{k-j,j}-1} \rangle$.}
            
        \STATE{Again using $\dlg$ (\autoref{cl:size-bound-on-Ti}), show that $T_{i,j+1}:=\partial_{z}(T_{i,j}/T_{k-j,j})$ has small $(\PSU{}{}/\PSU{}{}) \cdot (\SWSU{}{}{}/\SWSU{}{}{})$-circuit over $R_{j+1}(\x)$. So call the algorithm on $\sum_{i \in [k-j-1]}\, T_{i,j+1} \stackrel{?}{=}0$.}
        
        \STATE{$j \leftarrow j+1$.}
    \ENDFOR
    
    \STATE{At the end, $j=k-1$. Do PIT (\autoref{lem:swsu-hs}) on the single $(\PSU{}{}/\PSU{}{})\cdot (\SWSU{}{}{}/\SWSU{}{}{})$ circuit, over $R_{k-1}(\x)$. If it is zero, {\tt output} $0$ otherwise {\tt output} $1$.}
    \end{algorithmic}
    \end{algorithm}
    
\noindent \emph{Words of caution}: Throughout the algorithm there are intermediate expressions to be stored compactly. Think of them as `special' circuits in $\x$, but over the {\em function-field} $\F(\z)$. Keep track of their degrees wrt $z$;  and that of the sizes of their fractions represented in `bloated' circuit form. 
\section{Blacbox PIT for Depth-4 Circuits} \label{sec:pf-thm2}

We will give the proof of \autoref{thm:thm2} in this section. Before the details, we will state a few important definitions and lemmas from \cite{agrawal2016jacobian} to be referenced later.

\begin{definition}[Transcendence Degree]
    Polynomials $T_1,\hdots,T_m$ are called {\em algebraically dependent} if there exists a nonzero {\em annihilator} $A$ s.t.~$A(T_1,\hdots,T_m) = 0$. {\em Transcendence degree} is the size of the largest subset $S\subseteq \{T_1,\hdots,T_m\}$ that is algebraically independent. Then $S$ is called a {\em transcendence basis}.
\end{definition}

\begin{definition}[Faithful homomorphism]
    A homomorphism $\Phi : \F[\x] \to \F[\y]$ is faithful for $\T$ if $\tdeg_{\F}(\T) = \tdeg_{\F}(\Phi(\T))$.
\end{definition}
    
The reason for interest in faithful maps is due its usefulness in preserving the identity as shown in the following fact.

\begin{fact}[Theorem 2.4 \cite{agrawal2016jacobian}]\label{fact:faithful-is-useful}
    For any $C \in \F[y_1,\hdots,y_m]$, $C(\T) = 0 \iff C(\Phi(\T))=0$. 
\end{fact}

Here is an important criterion about the jacobian matrix which basically shows that it \emph{preserves} algabraic independence.
    
\begin{fact}[Jacobian criterion] \label{fact:jacobian-criterion}
    Let ${\bf f} \subset \F[\x]$ be a finite set of polynomials of degree at most $d$, and $\tdeg_{\F} ({\bf f}) \le r$. If char$(\F) = 0$, or char$(\F) > d^r$, then $\tdeg_{\F} ({\bf f}) =
    \rk_{\F(x)} \calJ_{\x}({\bf f})$.
\end{fact}

Jacobian criterion together with faithful maps give a recipe to design a map which drastically reduces number of variables, if trdeg is small.

\begin{lemma}[Lemma 2.7 \cite{agrawal2016jacobian}] \label{lem:faithful}
    Let $\T \in \F[\x]$ be a finite set of polynomials
    of degree at most $d$ and $\tdeg_{\F}(\T) \le r$, and char(F)=$0$ or $> d^r$. Let $\Psi' : \F[\x] \longrightarrow \F[z]$ such that $\rk_{\F(\x)} \calJ_{\x}(\T) = \rk_{\F(z)} \Psi'(\calJ_{\x}(\T))$. 
    
    Then, the map $\Phi : \F[\x] \longrightarrow \F[z,t,\y]$, such that $x_i \mapsto (\sum_{j\in [r]} y_j t^{ij}) + \Psi'(x_i)$, is a faithful homomorphism for $\T$.
\end{lemma}

In the next section we will use these tools to prove \autoref{thm:thm2}(b). The proof and calculations for \autoref{thm:thm2}(a) are very similar.

\subsection{PIT for \texorpdfstring{$\SPSP{[k]}{}{[\delta]}$}{SkPSPdelta}} \label{sec:SkPSPdelta}

We solve the PIT for a more general model than $\SPSP{[k]}{}{}$ by solving the following problem.

\begin{problem}\label{prob:prob1-thm2}
 Let $\{T_i \,|\, i \in [m]\}$ be $\PSP{}{[\delta]}$ circuits of (syntactic) degree at most $d$ and size $s$. Let the transcendence degree of $T_i$'s, $\tdeg_{\F}(T_1,\hdots,T_m)= k\ll s$. Further, $C(x_1,\hdots,x_m)$ be a circuit of $(\size + \deg) < s'$. Design a blackbox-PIT algorithm for $C(T_1,\hdots,T_m)$.
\end{problem}

Trivially, $\SPSP{[k]}{}{[\delta]}$ is a very special case of the above setting. Let $\T:=\{T_1,\hdots,T_m\}$. Let $\T_{k}:=\{T_1,\hdots,T_k\}$ be a transcendence basis. For $T_i=\prod_{j} g_{ij}$, we denote the set $L(T_i):=\{g_{ij}\,\mid\,j\}$. 

We want to find an explicit homomorphism $\Psi : \F[\x] \to \F[\x,z]$ s.t.~$\Psi(\calJ_{\x}(\T))$ is of a `nice' form. In the image we fix $\x$ suitably, to get a composed map $\Psi' : \F[\x] \longrightarrow \F[z]$ s.t.~$\rk_{\F(\x)}\calJ_{\x}(\T) = \rk_{\F(z)}\Psi'(\calJ_{\x}(\T))$. Then, we can extend this map to $\Phi : \F[\x] \longrightarrow \F[z,\y,t]$ s.t.~$x_i \mapsto (\sum_{j=1}^k\,y_jt^{ij}) + \Psi'(x_i)$, which is \emph{faithful} \autoref{lem:faithful}. We show that the map $\Phi$ can be efficiently constructed using a scaling and shifting map ($\Psi$) which is eventually fixed by the hitting set ($H'$ defining $\Psi'$) of a $\SWSP{}{}{[\delta]}$ circuit. Overall, $\Phi(f)$ is a $k+2$-variate polynomial for which a trivial hitting set exists.

Wlog, $\calJ_{\x}(\T)$ is full rank with respect to the variable set $\x_{k}=(x_1,\hdots,x_k)$. Thus, by assumption, $J_{{\x}_k}(\T_k) \ne 0$ (for notation, see section \ref{sec:prel}). We want to construct a $\Psi$ s.t.~$\Psi(J_{{\x}_k}(\T_k))$ has an `easier' PIT. We have the following identity \cite[Eqn.~3.1]{agrawal2016jacobian}, from the linearity of the determinant, and the simple observation that $\partial_{x}(T_i) \;=\; T_i \cdot \left(\sum_{j}\,\partial_{x}(g_{ij})/g_{ij}\right)$, where $T_i\,=\,\prod_{j}\,g_{ij} $:
\begin{align}
    J_{{\x}_k}(\T_k)\;=\;\sum_{g_1 \in L(T_1) , \hdots , g_k \in L(T_k)}\;\; \left( \frac{T_1\hdots T_k}{g_1 \hdots g_k}\right) \cdot J_{{\x}_k}(g_1,\hdots,g_k)\;. \label{jacobian-main-eq}
\end{align}

\textbf{The homomorphism $\Psi$.} To ensure the invertibility of all $g \in \bigcup_{i}\,L(T_i)$ we proceed as in section \ref{sec:pf:thm1}. Consider 
\begin{align*}
    h := \prod_{i \in [k]} \prod_{g \in L(T_i)} g = \prod_{i \in [\ell]} g,
\end{align*}
where $g \in \bigcup_{i}\,L(T_i)$ and $\ell \leq k \cdot s$. Note that $\deg h \leq d \cdot k \cdot s$ and $h$ is computable by $\PSP{}{}$ circuit of size $O(s)$. \autoref{lem:prod-sparse-hs} gives the relevant hitting set $\calH \subseteq \F^n$ which contains an evaluation point $\a = (a_1, \dots, a_n)$ such that $h(\a) \neq 0$ implying $g(\a) \neq 0$, for all $g \in \bigcup_{i}\,L(T_i)$. We emphasise that, unlike the previous case, here in the blackbox setting, we {\em do not} have individual access of $g$ to verify for the correct $\a$. Thus, we try out all $\a \in \calH$ to see whichever works. If the input polynomial $f$ is non-zero, then one such $\a$ must exist. This search adds a multiplicative blowup of $\poly(s)$, since the size of $\calH$ is $\poly(s)$. 

Fix an $\a = (a_1, \cdots, a_n) \in \calH$ and define $\Psi: \F[\x] \to \F[\x,z]$ as $x_i \mapsto z \cdot x_i + a_i$.
Denote the ring $\ringR[\x]$ where $\ringR:=\F[z]/\langle z^D \rangle$, and $D:=k\cdot(d-1)+1$. Being 1-1, $\Psi$ is clearly a non-zero preserving map. Moreover,
\begin{claim}\label{cl:psi-iff}
$J_{{\x}_k}(\T_k)\;=\; 0\; \iff\; \Psi(J_{{\x}_k}(\T_k))\;=\;0 $, over $\ringR[\x]$. 
\end{claim}

\begin{proof}
As $\d(T_i) \le d$, each entry of the matrix can be of degree at most $d-1$; therefore $\d(J_{{\x}_k}(\T_k)) \le  k(d-1)= D-1$. Thus, $\d_{z}(\Psi(J_{{\x}_k}(\T_k))) <D$. Hence, the conclusion.
\end{proof}
\autoref{jacobian-main-eq} implies that 
\begin{align}
    \Psi(J_{{\x}_k}(\T_k))\;=\;\Psi(T_1 \cdots T_k) \cdot \sum_{g_1 \in L(T_1) , \hdots , g_k \in L(T_k)}\;\; \frac{ \Psi(J_{{\x}_k}(g_1,\hdots,g_k))}{\Psi(g_1 \hdots g_k)} \;. \label{jacobian-main-eq-psi}
\end{align}
As $T_i$ has product fanin $s$, the top-fanin in the sum in \autoref{jacobian-main-eq-psi} can be at most $s^k$. Then define, 
\begin{align}
\widetilde{F}\;:=\;\sum_{g_1 \in L(T_1) ,\hdots , g_k \in L(T_k)}\;\; \frac{\Psi(J_{{\x}_k}(g_1,\hdots,g_k))}{\Psi(g_1 \hdots g_k)} \;,\;\text{ over }\,\ringR[\x]. \label{eq:main-eq-F-psi}
\end{align}

\textbf{Well-definability of $\widetilde{F}$.} Note that, \[\Psi(g_i) \bmod z \ne 0 \implies 1/\Psi(g_1 \cdots g_k) \in \F[[\x,z]].\] Thus, RHS is an element in $\F[[\x,z]]$ and taking $\bmod\, z^D$ it is in $\ringR[\x]$. We remark that instead of minimally reducing $\bmod\,z^D$, we will work with an $F \in \F[z,\x]$ such that    $F=\tilde{F}$ over $\ringR[\x]$. Further, we ensure that the degree of $z$ is polynomially bounded.

\begin{claim}\label{cl:psi-to-F}
Over $\ringR[\x]$, $\Psi(J_{{\x}_k}(\T_k)) = 0  \iff F =0$.  
\end{claim}

\begin{proof}[Proof sketch]
This follows from the invertibility of $\Psi(T_1 \cdots T_k)$ in $R[\x]$.
\end{proof}

\textbf{The hitting set $H'$.} By $J_{\x_k}(\T_k) \ne 0$, and Claims \ref{cl:psi-iff}-\ref{cl:psi-to-F}, we have $F \ne 0$ over $\ringR[\x]$. We want to find $H' \subseteq \F^{n}$, s.t.~$\Psi(J_{\x_k}(\T_k))\rvert_{\x=\a} \ne 0$, for some $\a \in H'$ (which will ensure the rank-preservation). 
Towards this, we will show (below) that $F$ has $s^{O(\delta k)}$-size $\SWSP{}{}{[\delta]}$-circuit over $\ringR[\x]$. Next, \autoref{thm:swsp-hs} provides the hitting set $H'$ in time $s^{O(\delta^2 k \log s)}$.


\begin{claim}[Main size bound]\label{claim:small-size-F-swsp}
$F \in \ringR[\x]$ has $\SWSP{}{}{[\delta]}$-circuit of size $ (s3^{\delta})^{O(k)}$.
\end{claim}

\noindent The proof studies the two parts of \autoref{eq:main-eq-F-psi}---
\begin{enumerate}
    \item The numerator $\Psi(J_{{\x}_k}(g_1,\hdots,g_k))$ has $O(3^{\delta}2^kk!ks)$-size $\SWSP{}{}{[\delta-1]}$-circuit (see \autoref{lem:num-size-thm2}), and
    \item $1/\Psi(g_1 \cdots g_k)$, for $g_i \in L(T_i)$ has $(s3^{\delta})^{O(k)}$-size $\SWSP{}{}{[\delta]}$-circuit; both over $\ringR[\x]$ (see \autoref{lem:denom-size-thm2}).
\end{enumerate}

\noindent We need the following two claims to prove the numerator size bound.

\begin{claim}\label{claim:jacobiansize}
    Let $g_i \in L(T_i)$, where $T_i \in \PSP{}{[\delta]}$ of size atmost $s$, then the polynomial $J_{{\x}_k}(g_1,\hdots,g_k)$ is computable by $\SPSP{[k!]}{[k]}{[\delta-1]}$ of size $O(k!\,ks)$.
\end{claim}
    
\begin{proof-sketch}
    Each entry of the matrix has degree at most $\delta-1$. Trivial expansion gives $k!$ top-fanin where each product (of fanin $k$) has size $\sum_{i}\,\size(g_i)$. As, $\size(T_i) \le s$, trivially each $\size(g_i) \le s$. Therefore, the total size is $k! \cdot\sum_{i}\,\size(g_i)=O(k!\,ks)$.
\end{proof-sketch}
    
\begin{claim} \label{claim:size-g}
    Let $g \in \Sigma \Pi^{\delta}$, then $\Psi(g) \in \Sigma \Pi^{\delta}$ of size $3^{\delta} \cdot \size(g)$ (for $n\gg\delta$).
\end{claim}
\begin{proof-sketch}
    Each monomial $\x^{\e}$ of degree $\delta$, can produce $\prod_i (e_i+1) \le ((\sum_{i} e_i +n)/n)^n \le (\delta/n+1)^n$-many monomials, by AM-GM inequality as $\sum_{i} e_i \le \delta$. As $\delta/ n \rightarrow 0$, we have $(1+\delta/n)^n \rightarrow e^{\delta}$. As $e < 3$, the upper bound follows.
\end{proof-sketch}

\begin{lemma}[Numerator size] \label{lem:num-size-thm2}
$\Psi(J_{{\x}_k}(g_1,\hdots,g_k))$ is computable by $\SWSP{}{}{[\delta-1]}$ of size $O(3^{\delta}\,2^kk\,k!s)=:s_2$.
\end{lemma}

\begin{proof}
    In \autoref{claim:jacobiansize} we showed that $J_{{\x}_k}(g_1,\hdots,g_k) \in \SPSP{[k!]}{[k]}{[\delta-1]}$ of size $O(k!ks)$. Moreover, for a $g \in \Sigma \Pi^{[\delta-1]}$, we have $\Psi(g) \in \Sigma \Pi^{[\delta-1]}$ of size at most $3^{\delta} \cdot \size(g)$, over $\ringR[\x]$ due to \autoref{claim:size-g}).
    
    Combining these, one concludes that $\Psi(J_{{\x}_k}(g_1,\hdots,g_k)) \in \SPSP{[k!]}{[k]}{[\delta-1]}$, of size $O(3^{\delta}\,k!ks)$. We \emph{convert} the $\Pi$-gate to $\wedge$ gate using waring identity (\autoref{lem:waring-rank}) which blowsup the size by a multiple of $2^{k-1}$. Thus, $\Psi(J_{{\x}_k}(g_1,\hdots,g_k)) \in \SWSP{}{}{[\delta-1]}$ of size $O(3^{\delta}\,2^kk\,k!s)$.
\end{proof}

In the following lemma, using power series expansion of expressions like $1/(1-a\cdot z)$, we conclude that $1/\Psi(g)$ has a small $\SWSP{}{}{[\delta]}$-circuit, which would further imply the same for $1/\Psi(g_1\cdots g_k)$.

\begin{lemma}[Denominator size] \label{lem:denom-size-thm2}
Let $g_i \in L(T_i)$. Then, $1/\Psi(g_1\cdots g_k)$ can be computed by a $\SWSP{}{}{[\delta]}$-circuit of size $s_1:=(s3^{\delta})^{O(k)}$, over $\ringR[\x]$.
\end{lemma}

\begin{proof}
Let $g \in L(T_i)$ for some $i$. Assume, $\Psi(g)=A-z\cdot B$, for some $A \in \F$ and $B \in \ringR[\x]$ of degree $\delta$, with $\size(B) \le 3^{\delta} \cdot s$, from \autoref{claim:size-g}. Note that, over $\ringR[\x]$,
\begin{align}
    \frac{1}{\Psi(g)} \;=\; \frac{1}{A(1- \frac{B}{A}\cdot z)}  \;=\; \frac{1}{A} \cdot \sum_{i=0}^{D-1} \left(\frac{B}{A}\right)^{i} \cdot z^i\;.\label{eq:powerseries-1/g-psi}
\end{align}
As, $B^i$ has a trivial $\WSP{}{[\delta]}$-circuit (over $\ringR[\x]$) of size $\le 3^{\delta}\cdot s + i$; summing over $i \in [D-1]$, the overall size is at most $D \cdot 3^{\delta} \cdot s + O(D^2)$. As $D < k \cdot d$, we conclude that $1/\Psi(g)$ has $\SWSP{}{}{[\delta]}$ of size $\poly(s \cdot k \cdot d 3^{\delta})$, over $\ringR[\x]$. Multiplying $k$-many such products directly gives an upper bound of $(s\cdot 3^{\delta})^{O(k)}$, using \autoref{lem:prod-swsw} (basically, waring identity).
\end{proof}

\begin{proof}[Proof of \autoref{claim:small-size-F-swsp}]
Combining Lemmas \ref{lem:num-size-thm2}-\ref{lem:denom-size-thm2}, observe that $\Psi(J_{\x_k}(\cdot)/\Psi(\cdot)$ has $\SWSP{}{}{[\delta]}$-circuit of size at most $(s_1 \cdot s_2)^2 = (s \cdot 3^{\delta})^{O(k)}$, over $\ringR[\x]$, using \autoref{lem:prod-swsw}. Summing up at most $s^k$ many terms (by defn.~of $F$), the size still remains $(s\cdot 3^{\delta})^{O(k)}$.
\end{proof}

{\bf \noindent Degree bound.} As, syntactic degree of $T_i$ are bounded by $d$, and $\Psi$ maintain $\d_{\x}=\d_{z}$, we must have $\d_{z}(\Psi(J_{\x_k}(g_1,\hdots,g_k)) =\d_{\x}(J_{\x_k}(g_1,\hdots,g_k)) \le D-1$. 
Note that, \autoref{lem:num-size-thm2} actually works over $\F[\x,z]$ and thus there is no additional degree-blow up (in $z$). However, there is some degree blowup in \autoref{lem:denom-size-thm2}, due to \autoref{eq:powerseries-1/g-psi}.

Note that \autoref{eq:powerseries-1/g-psi} shows that over $\ringR[\x]$,  \[\frac{1}{\Psi(g)} =\left( \frac{1}{A^D} \right)\,\cdot \left(\sum_{i=0}^{D-1} A^{D-1-i}z^{i} \cdot B^i\right) =: \frac{p(\x,z)}{q},\] where $q= A^{D}$. We think of $p \in \F[\x,z]$ and $q \in \F$. 
Note, $\d_{z}(\Psi(g)) \le \delta$ implies $\d_{z}(p) \le \d_{z}((B\,z)^{D-1}) \le \delta \cdot (D-1)$. 

Finally, denote $1/\Psi(g_1 \cdots g_k)=: P_{g_1,\hdots,g_k}/Q_{g_1,\hdots,g_k}$, over $\ringR[\x]$. This is just multiplying $k$-many $(p/q)$'s; implying a degree blowup by a multiple of $k$. In particular -- $\d_{z}(P_{(\cdot)}) \le \delta \cdot k \cdot (D-1)$
Thus, in \autoref{eq:main-eq-F-psi}, summing up $s^k$-many terms gives an expression (over $\ringR[\x]$):
\[ F \;=\; \sum_{g_1 \in L(T_1),\hdots, g_k \in L(T_k)}\; \Psi(J_{\x_k}(g_1,\hdots,g_k)) \cdot \left( \frac{P_{g_1,\hdots,g_k}}{Q_{g_1,\hdots,g_k}} \right)\;=:\; \frac{P(\x,z)}{Q}\;.\]

Verify that $Q \in \F$.
The degree of $z$ also remains bounded by \[\max_{g_i \in L(T_i), i\in[k]} \d_{z}(P_{g_1,\hdots,g_k})+\delta k \le \poly(s).\]
Using the degree bounds, we finally have $P \in \F[\x,z]$ as a $\SWSP{}{}{[\delta]}$-circuit (over $\F(z)$) of size $n^{O(\delta)}\,(s3^{\delta})^{O(k)}$ $=3^{O(\delta k)} s^{O(k+\delta)} =: s_3$. 

We want to \emph{construct} a set $H' \subseteq \F^n$ such that the action $P(H',z) \ne 0$.  Using \cite{forbes2015deterministic} (\autoref{thm:swsp-hs}), we conclude that it has $s^{O(\delta\log s_3)}$ $= s^{O(\delta^2 k \log s)}$ size hitting set which is constructible in a similar time. 
Hence, the construction of $\Phi$ follows, making $\Phi(f)$ a $k+2$ variate polynomial. Finally, by the obvious degree bounds of $\y,z,t$ from the definition of $\Phi$, we get the blackbox PIT algorithm with time-complexity $s^{O(\delta^2 k \log s)}$; finishing \autoref{thm:thm2}(b).

We could also give the final hitting set for the general problem.

\begin{proof}[Solution to \autoref{prob:prob1-thm2}]
We know that \[C(T_1,\hdots,T_m) = 0 \iff E:= \Phi(C(T_1,\hdots,T_m)) =0.\] 
 Since, $H'$ can be constructed in $s^{O(\delta^2\,k\,\log s)}$-time, it is trivial to find hitting set for $E\rvert_{H'}$ (which is just a $k+2$-variate polynomial with the aformentioned degree bounds). The final hitting set for $E$ can be constructed in $s'^{O(k)} \cdot s^{O(\delta^2\,k\,\log s)}$-time.
\end{proof}


\begin{remark*}
\begin{enumerate}
    \setlength\itemsep{1em}
    \item As Jacobian Criterion (\autoref{fact:jacobian-criterion}) holds when the characteristic is $> d^{\,\tdeg}$, it is easy to conclude that our theorem holds for all fields of char $> d^k$.
    \item The above proof gives an efficient reduction from blackbox PIT for $\SPSP{[k]}{}{[\delta]}$ circuits to $\SWSP{}{}{[\delta]}$ circuits. In particular, a poly-time hitting set for $\SWSP{}{}{[\delta]}$ circuits would put PIT for $\SPSP{[k]}{}{[\delta]}$ in $\P$. 
    \item Also, $\DiDI$-technique (of Theorem \autoref{thm:thm1}) directly gives a blackbox algorithm, but the complexity is {\em exponentially} worse (in terms of $k$ in the exponent) for its recursive blowups.
\end{enumerate}
\end{remark*}

\subsection{PIT for \texorpdfstring{$\SPSU{[k]}{}{}$}{SkPSU}}

As we remarked earlier, the proof of \autoref{thm:thm2}(a) is similar to the one we discussed in section \ref{sec:SkPSPdelta}. Here we sketch the proof, stating some relevant changes. Similar to \autoref{thm:thm2}(b), we generalize this theorem and prove for a much bigger class of polynomials.  

\begin{problem}\label{prob:prob2-thm2}
 Let $\{T_i \,|\, i \in [m]\}$ be $\PSU{}{}$ circuits of (syntactic) degree at most $d$ and size $s$. Let the transcendence degree of $T_i$'s, $\tdeg_{\F}(T_1,\hdots,T_m) =:k \ll s$. Further, $C(x_1,\hdots,x_m)$ be a circuit of size + degree $<s'$. Design a blackbox-PIT algorithm for $C(T_1,\hdots,T_m)$.
\end{problem}

It is trivial to see that $\SPSU{[k]}{}{}$ is a very special case of the above settings. We will use the same idea (\& notation) as in \autoref{thm:thm2}(b), using the Jacobian technique. The main idea is to come up with $\Psi$ map, and correspondingly the hitting set $H'$. If $g \in L(T_i)$, then $\size(g) \le O(dn)$. 
The $D$ (and hence $R[\x]$) remains as before. Claims \ref{cl:psi-iff}-\ref{cl:psi-to-F} hold similarly. We will construct the hitting set $H'$ by showing that $F$ has a small $\SWSU{}{}{}$ circuit over $R[\x]$. 

Note that, \autoref{claim:jacobiansize} remains the same for $\SWSU{}{}{}$ (implying the same size blowup). However, \autoref{claim:size-g}, the size blowup is $O(d\,\size(g))$, because each monomial $x^e$ can only produce $d+1$ many monomials. Therefore, similar to \autoref{lem:denom-size-thm2}, one can show that $\Psi(J_{\x_k}(g_1,\hdots,g_k)) \in \SWSU{}{}{}$, of size $O(2^kk!kds)$. Similarly, the size in \autoref{lem:num-size-thm2} can be replaced by $s^{O(k)}$. Therefore, we get (similar to \autoref{claim:small-size-F-swsp}):

\begin{claim}
$F \in R[\x]$ has $\SWSU{}{}{}$-circuit of size $s^{O(k)}$.
\end{claim}

Next, the degree bound also remains the same. Following the same footsteps, it is not hard to see that 
while degree bound on $z$ remains $\poly(ksd)$. Therefore, $P \in \F[\x,z]$ has $\SWSU{}{}{}$-circuit of size $s^{O(k)}$.

We want to \emph{construct} a set $H' \subseteq \F^n$ such that the action $P(H',z) \ne 0$.  By \autoref{lem:swsu-hs}, we conclude that it has $s^{O(k\,\log \log s)}$ size hitting set  which is constructible in a similar time. Hence, the construction of map $\Phi$ and the theorem follows (from $z$-degree bound).

\begin{proof}[Solution to \autoref{prob:prob2-thm2}]
We know that \[C(T_1,\hdots,T_m) = 0 \iff E:= \Phi(C(T_1,\hdots,T_m)) =0.\] 
Since, $H'$ can be constructed in $s^{O(k \log \log s)}$ time, it is trivial to find hitting set for $E\rvert_{H'}$ (which is just a $k+2$-variate polynomial with the aforementioned degree bounds). The final hitting set for $E$ can be constructed in $s'^{O(k)} \cdot s^{O(k \log \log s)}$ time.
\end{proof}

\section{Conclusion}

This work introduces the powerful $\DiDI$-technique and solves three open problems in PIT for depth-$4$ circuits, namely $\SPSP{[k]}{}{[\delta]}$ (blackbox) and $\SPSU{[k]}{}{}$ (both whitebox and blackbox). Here are some immediate questions of interest which require rigorous investigation.
\begin{enumerate}
	\item Can the exponent in \autoref{thm:thm1} be improved to $O(k)$? Currently, it is exponential in $k$.    
    \item Can we improve \autoref{thm:thm2}(b) to $s^{O(\log \log s)}$ (like in \autoref{thm:thm2}(a))?
    \item Can we design a polynomial-time PIT for $\SPSP{[k]}{}{[\delta]}$?
    \item Design a polynomial time PIT for $\SWSP{}{}{[\delta]}$ circuits (i.e.~unbounded top-fanin)? 
    \item Can we solve PIT for $\Sigma^{[k]}\Pi\Sigma\mathsf{M}_2$ circuits efficiently (polynomial/quasipolynomial-time), where \(\Sigma\mathsf{M}_2\) denotes bivariate polynomials?    
    \item Can we design an efficient PIT for rational functions of the form $\Sigma^{}\,(1/\Sigma\wedge\Sigma)$ or $\Sigma^{}\,(1/\Sigma\Pi)$ (for {\em un}bounded top-fanin)?
\end{enumerate}

\bibliographystyle{./bib/customurlbst/alphaurlpp}
\bibliography{./bib/bibliography}

\end{document}